\documentclass[a4paper,12pt,twoside,onecolumn]{book}

\usepackage[utf8]{inputenc}     
\usepackage{xspace} 
\usepackage{eurosym} 
\usepackage{ stmaryrd }
\usepackage{listings}
\usepackage[verbose,a4paper,portrait,twoside,twocolumn=false, left=2cm, top=2cm]{geometry}

\usepackage[T1]{fontenc}

\usepackage{todonotes}

\usepackage{url}

\usepackage[table]{xcolor}

\usepackage{pgf}
\usepackage{pgfcore}
\usepgfmodule{plot}
\usepgfmodule{shapes}
\usepgfmodule{decorations}
\usepackage{tikz}
\usetikzlibrary{arrows,patterns,plotmarks,er,3d,automata,backgrounds,topaths,trees,petri,mindmap}
\usepackage{calc}
\usepackage{graphicx}
\usepackage{array}
\usepackage{amsmath}
\usepackage{amssymb}
\usepackage{amsthm}
\DeclareMathAlphabet{\mathpzc}{OT1}{pzc}{m}{it}
\usepackage{mathrsfs} 

\usepackage{fancyvrb}
\usepackage{proof}
\usepackage[bookmarks={true},bookmarksopen={true}]{hyperref}
\usepackage{cleveref}
\usepackage{tikz-qtree}
\usepackage{multicol}
\usepackage[normalem]{ulem}
\usepackage{amsmath}
\usepackage{makeidx}
\usepackage{fontawesome5}
\makeindex

\title{Lecture Notes on Verifying Graph Neural Networks}
\author{François {\sc Schwarzentruber} \\ ENS de Lyon}

\usepackage{ifthen}
\usepackage{adjustbox}
\usepackage{comment}
\usepackage{multicol}
\usepackage{blkarray}
\usepackage{bussproofs}

\usetikzlibrary{decorations.markings}
\usetikzlibrary{calc,patterns,angles,quotes}

\newcommand\indexalert\emph

\newcommand{\cadregris}[2]{
	
	\begin{center}
		\colorbox{gray!20!white}{\begin{minipage}{0.9\textwidth}\textbf{#1} ~ \\[-2mm]
				
				#2\end{minipage}}\end{center}}

\newcommand{\atmset}{\ensuremath{\mathit{AP}}\xspace}

\protected\def\MSO{\ifmmode \mbox{\sc MSO} \else {\sc MSO}\xspace\fi}
\protected\def\FO{\ifmmode \mbox{\sc FO} \else {\sc FO}\xspace\fi}
\protected\def\MMSO{\ifmmode {\mbox{\sc M}\MSO} \else {{\sc M}\MSO}\xspace\fi}
\protected\def\MFO{\ifmmode \mbox{\sc MFO} \else {\sc MFO}\xspace\fi}

\newcommand{\union}{\cup}

\newcommand{\bigunion}{\bigcup}

\newcommand{\compl}[1]{{\overline{{#1}}}}

\newcommand{\inter}{\cap}

\newcommand{\card}[1]{{\emph{card}(#1)}}

\newcommand{\set}[1]{\left\{#1\right\}}
\newcommand\multiset[1]{\{\!\!\{#1\}\!\!\}}

\newcommand{\suchthat}{\mid}

\newcommand{\bigand}{\bigwedge}
\newcommand{\lbigand}{\bigand}

\newcommand{\bigor}{\bigvee}
\newcommand{\lbigor}{\bigor}
\newcommand{\true}{\top}

\newcommand{\limply}{\rightarrow}

\newcommand{\bottom}{\bot}
\newcommand{\lequiv}{\leftrightarrow}

\renewcommand{\phi}{\varphi}

\newcommand{\lbox}{\square}

\newcommand{\ldiamond}{\Diamond}

\newcommand{\logiclanguage}[1]{{\ensuremath{\mathcal{L}_{#1}}\xspace}}

\newcommand{\grammaris}{\hspace{2mm}::=\hspace{2mm}}
\newcommand{\grammarseparation}{\hspace{2mm}\mid\hspace{2mm}}
\newcommand{\grammarsep}{\grammarseparation}

\newcommand{\setN}{\mathbb{N}}

\newcommand{\ensN}{\setN}

\newcommand{\setZ}{\mathbb{Z}}

\newcommand{\ensZ}{\setZ}

\newcommand{\setR}{\mathbb{R}}

\newcommand{\ensR}{\setR}

\newcommand{\propP}{\mathcal P}

\newtheorem{theorem}{Theorem}
\newtheorem{lemma}[theorem]{Lemma}

\newtheorem{fact}[theorem]{Fact}

\newtheorem{exercise}{Exercise}

\newtheorem{remark}[theorem]{Remark}
\newtheorem{proposition}[theorem]{Proposition}

\newtheorem{property}[theorem]{Property}

\newtheorem{example}[theorem]{Example}

\newtheorem{definition}{Definition}
\newtheorem{corollary}[theorem]{Corollary}

\newcommand{\algofunction}{\textbf{function }}

\newcommand{\algoprocedure}{\textbf{procedure }}

\newcommand{\algofor}{\textbf{for }}
\newcommand{\algodo}{\textbf{do }}

\definecolor{algocommentbackgroundcolor}{rgb}{1,1,0.5}

\newcommand{\algowhile}{\textbf{while }}

\newcommand{\algoif}{\textbf{if }}
\newcommand{\algothen}{\textbf{then }}
\newcommand{\algoelse}{\textbf{else }}

\newcommand{\algomatch}{\textbf{match }}

\newcommand{\algocase}{\textbf{case }}

\newcommand{\algoreturn}{\textbf{return }}
\newcommand{\algochoose}{\textbf{choose }}

\newcommand{\algoreject}{\textbf{reject }}

\newlength{\algoindentlongueur}
\newcommand{\algoindent}{\hspace*{\algoindentlongueur}}

\newlength{\algoindentavantvrulelongueur}
\newcommand{\algoindentavantvrule}{\hspace*{\algoindentavantvrulelongueur}}

\newlength{\dummy}

\newsavebox{\frameminipageboiteavecunnomsuperlongdesortequonnepuissepaslereutiliser}
\newenvironment{frameminipage}[2][c]{%
\begin{lrbox}{\frameminipageboiteavecunnomsuperlongdesortequonnepuissepaslereutiliser}%
\begin{minipage}[#1]{#2}%
} {%
\end{minipage}%
\end{lrbox}%
\framebox{\usebox{\frameminipageboiteavecunnomsuperlongdesortequonnepuissepaslereutiliser}}%
}

\newenvironment{algo} {
	\noindent
  \begin{frameminipage}{\linewidth}
} {
  \end{frameminipage}
}

\newenvironment{algobloc}{\setlength{\dummy}{\linewidth}\addtolength{\dummy}{- \algoindentlongueur}\addtolength{\dummy}{- \algoindentavantvrulelongueur}\algoindentavantvrule\vrule\algoindent\begin{minipage}{\dummy}}{\end{minipage}}

\newenvironment{algoblocprocedure}[1]
{\algoprocedure #1 \\  \begin{algobloc}}
	{\end{algobloc}}

\newenvironment{algoblocfunction}[1]
{\algofunction #1 \\  \begin{algobloc}}
{\end{algobloc}}

\newenvironment{algoblocmainfunction}[1]
{\textbf{main} \algofunction #1 \\  \begin{algobloc}}
{\end{algobloc}}

\newenvironment{algoblocfor}[1]
{\algofor #1 \algodo \\  \begin{algobloc}}
{\end{algobloc}}

\newcommand{\indexemph}[1]{\emph{#1}\index{#1}}

\tikzstyle{zoneprogramcounter} = [fill=gray!40, draw=none]
\tikzstyle{zonedatacounter} = [fill=gray!40, draw=none, decorate, decoration={snake, amplitude=1pt, segment length=4pt}]
\tikzstyle{portion} = [densely dotted, shade, top color = white, bottom color = gray!40]
\tikzstyle{portiondata} = [densely dotted, shade, top color = white, bottom color = gray!40, decorate, decoration={snake, amplitude=1pt, segment length=4pt}]
\tikzset{
	double arrow/.style args={#1 colored by #2 and #3}{
		-stealth,line width=#1,#2, 
		postaction={draw,-triangle 90 cap,#3,line width=(#1)-4pt,
			shorten <=4pt,shorten >=4}, 
	}
}
\tikzstyle{tapearrow} = [double arrow=10pt colored by black!50!white and black!20!white, -triangle 90 cap, fill = white]%
\tikzstyle{execarrow} = [line width=1pt,->]%

\tikzstyle{world}=[inner sep=0.5mm]
\tikzstyle{event}=[fill=gray!30, inner sep=0.5mm]
\tikzstyle{realworldarrowfromleft} = [initial left, initial text={}]
\tikzstyle{realworldarrowfromright} = [initial right, initial text={}]

\protected\def\ZPP{\ifmmode \mbox{\sc ZPP} \else {\sc ZPP}\xspace\fi}
\protected\def\RP{\ifmmode \mbox{\sc RP} \else {\sc RP}\xspace\fi}
\protected\def\coRP{\ifmmode \mbox{\sc coRP} \else {\sc coRP}\xspace\fi}
\protected\def\DTIME{\ifmmode \mbox{\sc Dtime} \else {\sc Dtime}\xspace\fi}
\protected\def\NTIME{\ifmmode \mbox{\sc Ntime} \else {\sc Ntime}\xspace\fi}
\protected\def\DSPACE{\ifmmode \mbox{\sc Dspace} \else {\sc Dspace}\xspace\fi}
\protected\def\NSPACE{\ifmmode \mbox{\sc Nspace} \else {\sc Nspace}\xspace\fi}
\protected\def\NP{\ifmmode \mbox{\sc NP} \else {\sc NP}\xspace\fi}
\protected\def\coNP{\ifmmode \mbox{\sc coNP} \else {\sc coNP}\xspace\fi}
\protected\def\NPSPACE{\ifmmode \mbox{\sc NPspace} \else {\sc NPspace}\xspace\fi}
\protected\def\PSPACE{\ifmmode \mbox{\sc Pspace} \else {\sc Pspace}\xspace\fi}
\protected\def\EXPSPACE{\ifmmode \mbox{\sc Expspace} \else {\sc Expspace}\xspace\fi}
\protected\def\TWOEXPSPACE{\ifmmode \mbox{\sc 2Expspace} \else {\sc 2Expspace}\xspace\fi}
\protected\def\PTIME{\ifmmode \mbox{\sc P} \else {\sc P}\xspace\fi}
\protected\def\NPTIME{\ifmmode \mbox{\sc NP} \else {\sc NP}\xspace\fi}
\protected\def\EXPTIME{\ifmmode \mbox{\sc Exptime} \else {\sc Exptime}\xspace\fi}
\protected\def\AEXPTIME{\ifmmode \mbox{\sc Aexptime} \else {\sc Aexptime}\xspace\fi}
\protected\def\NEXPTIME{\ifmmode \mbox{\sc NExptime} \else {\sc NExptime}\xspace\fi}
\protected\def\2EXPTIME{\ifmmode \mbox{\sc 2-Exptime} \else {\sc
		2-Exptime}\xspace\fi}
\DeclareRobustCommand{\kEXPTIME}[1][k]{\ifmmode \mbox{\sc $#1$-Exptime}
	\else {\sc $#1$-Exptime}\xspace\fi}
\DeclareRobustCommand{\kNEXPTIME}[1][k]{\ifmmode \mbox{\sc $#1$-NExptime}
	\else {\sc $#1$-NExptime}\xspace\fi}
\DeclareRobustCommand{\kEXPSPACE}[1][k]{\ifmmode \mbox{\sc $#1$-Expspace}
	\else {\sc $#1$-Expspace}\xspace\fi}
\protected\def\ELEMENTARY{\ifmmode \mbox{\sc Elementary} \else {\sc Elementary}\xspace\fi}
\protected\def\AEXPpol{\ifmmode \mbox{{\sc A}_{\text{pol}}\EXPTIME} \else
	{\sc A}$_{\text{pol}}$\EXPTIME\fi}
\protected\def\APTIME{\ifmmode \mbox{\sc Aptime} \else {\sc Aptime}\xspace\fi}
\protected\def\AEXPSPACE{\ifmmode \mbox{\sc Aexpspace} \else {\sc Aexpspace}\xspace\fi}

\tikzstyle{cell} = [draw,minimum height=5mm,minimum width=5mm]
\tikzset{
	cellcolor/.cd,
	0/.style={fill=gray!20!white},
	1/.style={fill=yellow!20!white}
}
\tikzstyle{cellalive} = [fill=yellow!20!white]

\newcommand\sigmanew{\sigma_\text{new}}

\newcommand\listeventsempty\epsilon

\newcommand{\initialword}{\omega}

\definecolor{C}{rgb}{1,1,1}
\definecolor{C }{rgb}{0.8,0.8,0.8}

\definecolor{Ca}{rgb}{0.95,1,0.5}
\definecolor{Cb}{rgb}{1,0.9,0.5}
\definecolor{Cw_2}{rgb}{0.95,1,0.5}
\definecolor{Cw_n}{rgb}{1,0.9,0.5}

\definecolor{Cq'-}{rgb}{0.5,0.5,1}
\definecolor{C-q_0}{rgb}{0.5,0.5,1}
\definecolor{C-q_1}{rgb}{0.5,0.4,1}
\definecolor{C-q'}{rgb}{0.5,0.8,0.5}
\definecolor{Cq_f-}{rgb}{0.4,0.9,0.5}

\definecolor{Cq_0,w_1}{rgb}{1,0.6,0.5}
\definecolor{Cq_0,}{rgb}{1,0.6,0.4}
\definecolor{Cq_2,}{rgb}{0.8,0.9,0.4}
\definecolor{Cq_2,a}{rgb}{0.8,0.9,0.4}
\definecolor{Cq_f,a}{rgb}{0.7,0.9,0.4}
\definecolor{Cq_0,a}{rgb}{1,0.6,0.5}
\definecolor{Cq_0,b}{rgb}{1,0.5,0.5}
\definecolor{Cq_1,a}{rgb}{1,0.5,0.4}
\definecolor{Cq_1,b}{rgb}{1,0.4,0.4}
\definecolor{Cq,a}{rgb}{1,0.5,0.5}
\definecolor{Cq',a}{rgb}{1,0.5,0.5}
\definecolor{Cq',b}{rgb}{1,0.5,0.5}
\definecolor{Cq_0,1}{rgb}{0.5,0.5,0.5}

\provideboolean{student}

\newcommand{\trou}[1]{\ifthenelse{\boolean{student}}
	{\colorbox{gray!10!white}{\phantom{#1}}}
	{#1}
}

\ifthenelse{\boolean{student}}
{\excludecomment{hidden}}
{}

\usepackage{mdframed}

\newmdenv{allfour}

\newmdenv[rightline=false,topline=false,bottomline=false]{solutionbox}

\ifthenelse{\boolean{student}}
{\excludecomment{solution}}
{
}

\ifthenelse{\boolean{student}}
{\excludecomment{examplehidden}}
{}

\ifthenelse{\boolean{student}}
{\excludecomment{exercicehidden}}
{}

\ifthenelse{\boolean{student}}
{\excludecomment{remarkhidden}}
{}

\definecolor{proba}{RGB}{0, 0, 128}

\newcommand\proba[1]{
	\colorlet{savedleftcolor}{.}
	\color{proba}
	\mathbb P
	(
	\color{savedleftcolor}
	#1\color{proba}
	)
	\color{savedleftcolor}}

\usepackage{tikzsymbols}

\newcommand{\taille}[1]{|#1|}
\newcommand\sizeof\taille

\newcommand\instancemot\initialword

\newcommand{\aGNN}{N}

\newcommand{\weightcolor}[1]{\textcolor{blue}{#1}}

\newcommand\columnvector[1]{
\left(
\begin{array}{c}
#1
\end{array}
\right)
}

\newcommand{\studentproject}[1]{
\begin{center}
\fbox{\faGraduationCap #1}
\end{center}
}

\newcommand{\modelM}{\mathcal M}
\newcommand\Ksharp{\ensuremath{\text{K}^\#}\xspace}

\begin{document}

\maketitle

%

\section*{Preface}

In these lecture notes, we first recall the connection between graph neural networks, Weisfeiler-Lehman tests and logics such as first-order logic and graded modal logic. We then present a modal logic in which counting modalities appear in linear inequalities in order to solve verification tasks on graph neural networks. We describe an algorithm for the satisfiability problem of that logic. It is inspired from the tableau method of vanilla modal logic, extended with reasoning in quantifier-free fragment Boolean algebra with Presburger arithmetic.

In order to give an idea to lecturers that may use this material, this course was taught at a master level at ENS de Lyon. The lecture was divided in four 2-hour sessions + mid-exam as follows:
\begin{itemize}
\item 29th of September 2025: definition of GNN, color refinement, FO, FO2, FOC, FOC2, FO2 is NEXPTIME-hard, correspondence between FOC2 and color refinement
\item 3th of October: discussion on the correspondence, ML, standard translation, GML, correspondence between GML and color refinement, tableau method for ML, satisfiability problem ML is in PSPACE. (soundness and completeness of the tableau method was not proven)
\item 6th of October: satisfiability problem ML is PSPACE-complete. GNN with truncReLU and \Ksharp. GNN with truncReLU and $\exists PA^*$ (only the statement).
\item 10th of October: GNN with truncReLU and $\exists PA^*$. Issue with the satisfiability procedure. QFBAPA and \Ksharp full satisfiability procedure. 

\item 13th of October: mid-exam, followed by a discussion session 
\end{itemize}

First I want to thank Pierre Nunn: we started to work on GNNs and logic in 2022.
I want to warmly thank Artem Chernobrovkin, Marco Sälzer and Nicolas Troquard for all the discussions that helped to improve these lecture notes. I also warmly thank the master students (M2 level) at ENS de Lyon for their comments during the lecture.
Enjoy!

\tableofcontents

\chapter{Graph neural networks and Weisfeiler-Lehman tests}
\fbox{
Key references: \cite{DBLP:journals/corr/abs-2003-04078}, \cite{DBLP:conf/lics/Grohe21}}

\section{Motivation}

\index{graph neural network}
Graph neural networks have been proposed as machine learning model in \cite{DBLP:journals/tnn/ScarselliGTHM0}. They are used in various applications:
\begin{itemize}
\item recommendation in social network \cite{DBLP:journals/kbs/SalamatLJ21},
\item knowledge graphs \cite{Zi22KG},
\item chemistry \cite{Reiser22GGNmaterialchemistry},
\item drug discovery \cite{XIONG20211382}
\item voice separation in music scores \cite{DBLP:conf/ijcai/KarystinaiosFW23}, etc.
\end{itemize}

They are used for classification (yes/no questions) or regression (guess a value) for a given graph or a pointed graph (a vertex in some graph)\footnote{In the literature, they may say `on graph-nodes` or `on nodes`.}. The following table gives some examples of informal questions:

\begin{center}
\begin{tabular}{l|p{55mm}|p{5cm}}
& classification & regression \\
\hline
on graphs & does the graph represent a toxic molecule? & what is the temperature of fusion of a given molecule represented by a graph?\\
on pointed graphs & do we recommend some person? & what is the price of some furniture?
\end{tabular}
\end{center}

\section{Graph neural networks}

\newcommand{\labelling}{\ell}
\newcommand{\labellinginitial}{\ell_0}
\newcommand{\numberoflayers}{L}
\newcommand{\sem}[2]{[\![#1]\!]_{#2}}
\newcommand{\semone}[1]{[\![#1]\!]}
\newcommand{\activationfunction}{\alpha}

\index{labelled graph}
\index{vertex}
\index{edge}
\index{graph}
We consider labelled directed graphs $G = (V, E, \labelling)$ where $V$ is a non-empty finite set of vertices, $E \subseteq V \times V$ is a set of edges and 
$\labelling : V \rightarrow \setR^d$ 
is a labelling function, i.e. each vertex $u$ is labelled with a vector $\labelling(u)$ containing $d$ real numbers. We denote by~$E(u)$ the set of successors of $u$. Formally, $E(u) = \set{v \in V \suchthat (u, v) \in E}$. We encode a standard \indexemph{Kripke structure} by a labelled graph $G = (V, E, \labellinginitial)$ with $\labellinginitial : V \rightarrow \set{0, 1}^d$.

\index{graph neural network}
\index{weight}
A GNN $\aGNN$ is usually defined as a tuple of parameters (see \cite{DBLP:conf/iclr/BarceloKM0RS20}, \cite{DBLP:conf/ijcai/NunnSST24}, \cite{DBLP:conf/icalp/BenediktLMT24}). To make it more concrete, we present it as an algorithm (parametrized by the weights computed during some learning process), see \Cref{figure:GNN_algorithm}. It takes as an input a pointed graph made up of a labelled graph $G = (V, E, \labellinginitial)$ and a vertex $u$. It outputs yes/no.

\begin{figure}[h]
\begin{center}
\begin{minipage}{80mm}
\begin{algo}
\begin{algoblocmainfunction}{$\aGNN((V, E, \labellinginitial), u)$}

$\labelling_1 := layer_1(V,E, \labelling_0)$

$\labelling_2 := layer_2(V,E, \labelling_1)$

$\vdots$

$\labelling_\numberoflayers := layer_\numberoflayers(V,E, \labelling_{\numberoflayers-1})$

\algoreturn yes \algoif $\weightcolor{w}^t \labelling_\numberoflayers(u) + \weightcolor{b} \geq 0$ \algoelse no

\end{algoblocmainfunction}
\end{algo}
\end{minipage}
\end{center}

\begin{center}
\begin{minipage}{12cm}
\begin{algo}
\begin{algoblocfunction}{$layer_i(V,E, \labelling)$ }

$\labelling'$ := new labelling $V \rightarrow \setR^d$

\begin{algoblocfor}{vertices $u \in V$}

$\labelling'[u] := \vec \activationfunction(\weightcolor{A_i} \times \labelling[u] + \weightcolor{B_i} \times \sum \multiset{\labelling[v] \suchthat v \in E(u)} + \weightcolor{b_i})$

\end{algoblocfor}

\algoreturn $\labelling'$

\end{algoblocfunction}
\end{algo}
\end{minipage}
\end{center}

\caption{A graph neural network $\aGNN$ presented as an algorithm.\label{figure:GNN_algorithm} The main function is $\aGNN$. It computes a sequence of labellings $\labelling_1, \labelling_2, \dots, \labelling_\numberoflayers$ via functions $layer_i$. Learnt weights are $\weightcolor{w}, \weightcolor{b}, \weightcolor{A_i}, \weightcolor{B_i}, \weightcolor{b_i}$.}
\end{figure}

\begin{definition}[graph neural network]
A GNN is an algorithm as shown in \Cref{figure:GNN_algorithm}.
\end{definition}

A GNN $\aGNN$ computes a sequence of labellings $\labelling_1, \labelling_2, \dots, \labelling_\numberoflayers$ via the application of so-called layers. For avoiding cumbersome notations, we suppose that all these labellings assign a vector of dimension $d$ to each vertex (while in generality the dimension $d$ may be different for each layer). 
In each layer $layer_i$, the function $\vec \activationfunction : \setR^d \rightarrow \setR^d$ is the point-wise application of an activation 
 function $\activationfunction : \setR \rightarrow \setR$. The \indexemph{activation function} $\activationfunction$ could be for instance $ReLU : x \mapsto \max(0, x)$
 or $truncReLU : x \mapsto \max(0, \min(x, 1))$.
 \index{truncated ReLU}
 \index{ReLU}
 Then $A_i \in \setR^{d \times d}$ and $B_i \in \setR^{d \times d}$ are matrices of weights, $\times$ is the standard matrix-vector multiplication, $b_i \in \setR^d$ is a bias vector,  $\multiset{.}$ is the multiset notation, and $\sum$ is the summation operation of all vectors in the multiset.

  The ending linear inequality $w^t \labelling_\numberoflayers(u) + b \geq 0$ where $w \in \ensR^d$ and $b \in \setR$ are weights is used to classify the vertex $u$.
In the sequel, the set of pointed labelled graphs positively classified by a GNN $\aGNN$ is denoted by
$\semone{\aGNN} := \set{(G, u) \suchthat \aGNN(G, u) \text{ returns } \text{yes}}.$

\newcommand{\aGNNoutlayer}[1]{N^{(#1)}}
\newcommand{\timestep}{t}

We can also present a GNN $\aGNN$ that computes a sequence $\aGNNoutlayer 0, \aGNNoutlayer 1, \dots$ of labellings:

\begin{itemize}
\item $\aGNNoutlayer{0}(G, u) = \labellinginitial(u)$;
\item $\aGNNoutlayer{t}(G,u) = \vec \activationfunction(\weightcolor{A_t} \times \aGNNoutlayer{t-1}(G,u) + \weightcolor{B_\timestep} \times \sum \multiset{\aGNNoutlayer{t-1}(G,v) \suchthat v \in E(u)} + \weightcolor{b_\timestep})$ for all $u \in V$.
\end{itemize}

\section{GNN on graphs}

For graph classification, the last operation could be:

\begin{center}
\algoreturn yes \algoif $\weightcolor{w}\sum_{u\in V} \labelling_{\numberoflayers}(u) + \weightcolor{b} \geq 0$ \algoelse no
\end{center}

\index{global readout}
The sum is taken over \emph{all} vertices $u$ in the graph. This kind of operation is \emph{global readout}. We could also have this operation at each layer:

\begin{align*}
\aGNNoutlayer{t}(G,u) = & \vec \activationfunction(\weightcolor{A_t} \times \aGNNoutlayer{t-1}(G,u) \\
 & + \weightcolor{B_\timestep} \times \sum \multiset{\aGNNoutlayer{t-1}(G,v) \suchthat v \in E(u)} \\
 & + \weightcolor{C_\timestep} \times \sum \multiset{\aGNNoutlayer{t-1}(G,v) \suchthat v \in V} \\
 & + \weightcolor{b_\timestep}).
\end{align*}

In order to keep the course simple, we postpone the discussion on global readout in Chapter~\ref{chapter:globalreadout}.

\section{Graph isomorphism}

\index{isomorphism}
Two graphs are isomorphic if they are the same, up to vertex renaming. Here is a formal definition.

\begin{example}Here are two graphs that isomorphic:
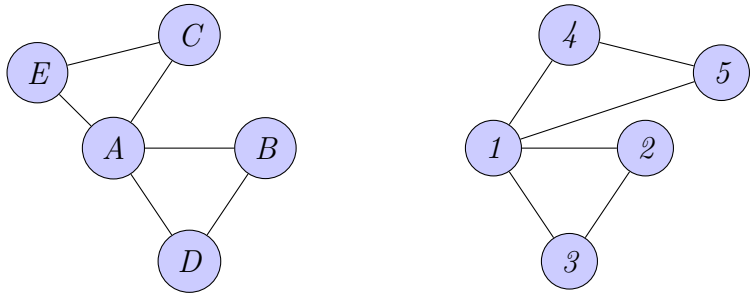
\begin{center}
\begin{tikzpicture}[scale=1, every node/.style={circle, draw, fill=blue!20}]
  \node (a) at (0,0) {A};
  \node (b) at (2,0) {B};
  \node (c) at (1,1.5) {C};
  \node (d) at (1,-1.5) {D};
  \node (e) at (-1,1) {E};

  \draw (a) -- (b);
  \draw (a) -- (c);
  \draw (a) -- (d);
  \draw (a) -- (e);
  \draw (c) -- (e);
  \draw (d) -- (b);

  \node (1) at (5,0) {1};
  \node (2) at (7,0) {2};
  \node (3) at (6,-1.5) {3};
  \node (4) at (6,1.5) {4};
  \node (5) at (8,1) {5};

  \draw (1) -- (2);
  \draw (1) -- (3);
  \draw (1) -- (4);
  \draw (1) -- (5);
  \draw (4) -- (5);
  \draw (3) -- (2);
\end{tikzpicture}
\end{center}
\end{example}

\newcommand{\isomorphism}{\pi}
\begin{definition}
	Two labelled graphs $G = (V, E, \labelling)$ and $G' = (V', E', \labelling')$ are \emph{isomorphic} if there exists a bijection $\isomorphism : V \rightarrow V'$ such that:
	\begin{enumerate}
	\item for all vertices $u \in V$, $\labelling[u] = \labelling[\isomorphism(u)]$
	\item for all vertices $u, v \in V$, $uEv$ iff $\isomorphism(u) E' \isomorphism(v)$.
	\end{enumerate}
\end{definition}

\index{NP}
Graph isomorphism is in NP, but likely to be NP-complete, because then the polynomial hierarchy would collapse \cite{DBLP:journals/jcss/Schoning88}.
Testing graph isomorphisms can be done in quasipolynomial time \cite{DBLP:conf/stoc/Babai16}: more precisely in time $exp((\log n)^{O(1)})$ where $n$ is the number of vertices.
Weisfeiler-Lehman tests are procedures running in poly-time\footnote{It is possible to implement it in $O((|V| + |E|) \log |V|)$ \cite{DBLP:journals/tcs/CardonC82}.} which do ``almost" solve graph isomorphism.

\section{1-WL aka Colour refinement}

\newcommand{\algocolorrefinement}{\textsf{cr}\xspace}
\newcommand{\algocolorrefinementattime}[1]{\textsf{cr}^{(#1)}}
\newcommand{\algocolorrefine}{\textsf{refine}\xspace}
\newcommand{\crtime}{t}

\index{colour refinement}
\index{Weisfeiler-Lehman}
\subsection{Description}

There are different presentations of 1-WL (1-Weisfeiler-Lehman) aka color refinement aka naive vertex refinement in the literature \cite{morgan1965generation}. Informally, it works as follows on a graph $G = (V,E, \labellinginitial)$.
\begin{itemize}
\item Initially, colour each vertex $v$ with the colour $\labellinginitial(v)$.
\item Iteratively, recolour each vertex based on its current colour and the multiset of colours of its neighbours.
\item Repeat until the colouring stabilizes.
\end{itemize}

Formally, we define a sequence $(\algocolorrefinementattime{0}(G), \algocolorrefinementattime{1}(G), \dots)$ of labellings as follows:

\begin{itemize}
\item $\algocolorrefinementattime{0}(G, u) = \labellinginitial(u)$ for all $u \in V$;
\item  $\algocolorrefinementattime{\crtime+1}(G, u) = \left(\algocolorrefinementattime{\crtime}(G, u), \multiset{\algocolorrefinementattime{\crtime}(G, v) \suchthat u E v}\right)$.
\end{itemize}

\newcommand{\finerthan}{\sqsubseteq}
We wrote $\algocolorrefinementattime{0}(G, u)$ instead of the cumbersome $\algocolorrefinementattime{0}(G)(u)$. In each round, the labelling gets finer: $\algocolorrefinementattime{\crtime+1}(G)$ is finer than $\algocolorrefinementattime{\crtime}(G)$. We use the notation $\finerthan$ to say `finer than`:

$$ \dots \finerthan \algocolorrefinementattime{2}(G) \finerthan 
\algocolorrefinementattime{1}(G) \finerthan 
\algocolorrefinementattime{0}(G)$$

  For some $\crtime_G < |V|$, we have that $\algocolorrefinementattime{\crtime+1}(G)$ and $\algocolorrefinementattime{\crtime}(G)$ are equivalent in the following sense.

\begin{figure}
\begin{center}
\begin{minipage}{75mm}
\begin{algo}
\begin{algoblocmainfunction}{$\algocolorrefinement(V, E, \labellinginitial)$}

$\crtime := 0$

\textbf{repeat}

\begin{algobloc}

$\labelling_{\crtime+1} := \algocolorrefine(V,E, \labelling_{\crtime})$

$\crtime := \crtime+1$

\end{algobloc}

\textbf{until} $\labelling_{\crtime+1}$ and $\labelling_{\crtime}$ are equivalent

\algoreturn $\labelling_{i}$

\end{algoblocmainfunction}
\end{algo}
\end{minipage}
\hfill
\begin{minipage}{80mm}
\begin{algo}
\begin{algoblocfunction}{$\algocolorrefine(V,E, \labelling)$ }

$\labelling'$ := new labelling 

\begin{algoblocfor}{vertices $u \in V$}

$\labelling'[u] := (\labelling[u], \multiset{\labelling[v] \suchthat v \in E(u)})$

\end{algoblocfor}

\algoreturn $\labelling'$

\end{algoblocfunction}
\end{algo}
\end{minipage}
\end{center}

\caption{Colour refinement algorithm.\label{algo:colorrefinement}}
\end{figure}

\begin{definition}
Two labellings $\labelling$ and $\labelling'$ are equivalent if 

\hfill for all vertices $u, v \in V$, $\labelling[u] = \labelling[v]$ iff $\labelling'[u] = \labelling'[v]$.
\end{definition}

\index{partition}
Said differently, $\algocolorrefinementattime{\crtime_G+1}(G)$ and $\algocolorrefinementattime{\crtime_G}(G)$ induce the same partition. This is used as a stopping condition in colour refinement. We write $\algocolorrefinement(G)$ instead of $\algocolorrefinementattime{\crtime_G+1}(G)$. This is the output of colour refinement. 
\Cref{algo:colorrefinement} gives the pseudo-code for colour refinement.

\subsection{Example}

\begin{exercise}(from \cite{DBLP:conf/lics/Grohe21})
Compute $\algocolorrefinement(G)$ and $\algocolorrefinement(G')$ for the graphs $G$ and $G'$ below:

\begin{center}
\begin{tikzpicture}[scale=1.2, every node/.style={circle, draw, fill=blue!30, inner sep=1.5pt}]
    \node (a) at (0,1) {};
    \node (b) at (-1,0.5) {};
    \node (c) at (-1,-0.5) {};
    \node (d) at (0,-1) {};
    \node (e) at (1,-0.5) {};
    \node (f) at (1,0.5) {};

    \draw (a)--(b)--(c)--(d)--(e)--(f)--(a);
    \draw (b)--(e);

    \node (g) at (4,1) {};
    \node (h) at (3,0.5) {};
    \node (i) at (5,0.5) {};
    \node (j) at (3,-0.5) {};
    \node (k) at (5,-0.5) {};
    \node (l) at (4,-1) {};

    \draw (g)--(h)--(i)--(g);
    \draw (j)--(k)--(l)--(j);
    \draw (h)--(j);
    \draw (h)--(k);

\end{tikzpicture}
\end{center}

\end{exercise}

\newcommand{\histogramalgorefinement}[1]{\multiset{\algocolorrefinement(#1)}}

\subsection{Implementation}

It is possible to compute the final partition corresponding to $\algocolorrefinement(G)$ in $O((|V| + |E|) \log |V|)$.

\studentproject{Read \cite{DBLP:journals/tcs/CardonC82}, also \cite{DBLP:journals/siamcomp/PaigeT87}.}

 The interested reader may also look at \cite{DBLP:journals/mst/BerkholzBG17} for tight complexity.

\subsection{Indistinguishability}

$G, G'$ are \algocolorrefinement-indistinguishable if $\algocolorrefinement(G)$ and  $\algocolorrefinement(G')$ have the \emph{same histogram of colors}: the number of vertices of a given color are the same via $\algocolorrefinement(G)$ and  $\algocolorrefinement(G')$. More formally, we define:

$$\histogramalgorefinement{G} := \multiset{\algocolorrefinement(G,u) \suchthat u \in V}$$

\begin{definition}
$G, G'$ are \algocolorrefinement-indistinguishable if $\histogramalgorefinement G = \histogramalgorefinement {G'}$.
\end{definition}

\begin{definition}
$G,u$ and  $G', u'$ are \algocolorrefinement-indistinguishable if $\algocolorrefinement( G,u) = \algocolorrefinement(G',u')$.
\end{definition}

\subsection{Link with isomorphism}

The following proposition says that the output of $\algocolorrefinement$ is the same for isomorphic graphs. We say that $\algocolorrefinement$ is an \emph{equivariant}.

\index{isomorphism}
\begin{proposition}
If $G$ and $G'$ are isomorphic \textbf{then} $G, G'$ are  \algocolorrefinement-indistinguable.
\end{proposition}

\begin{proof}
Suppose that $G$ and $G'$ are isomorphic.
 Let $\labelling_0, \labelling_1, \dots$ be the labellings of the execution of $\algocolorrefinement(G)$.
 Let $\labelling'_0, \labelling'_1, \dots$ be the labellings of the execution of $\algocolorrefinement(G')$.
 
 Let $\isomorphism$ an isomorphism from $G$ into $G'$. 
 Given $t \in \setN$, we consider the property $\propP(t)$: 
 \begin{center}
 for all $u \in V$, 
 $\algocolorrefinementattime{t}(G,u) = \algocolorrefinementattime{t}(G,\isomorphism(u)).$ 
 \end{center} 
 
 For $t = 0$, $\propP(0)$ holds by definition of an isomorphism.
 
Suppose $\propP(t)$ for some $t$. The computation is: 

\begin{align*}
\algocolorrefinementattime{t+1}(G,u) & := (\algocolorrefinementattime{t}(G,u), \multiset{\algocolorrefinementattime{t}(G,v) \mid v \in E(u)}) \\
& = (\algocolorrefinementattime{t}(G',\isomorphism(u)), \multiset{\algocolorrefinementattime{t}(G',\isomorphism(v)) \mid v \in E(u)})  \\
& = (\algocolorrefinementattime{t}(G',\isomorphism(u)), \multiset{\algocolorrefinementattime{t}(G',\isomorphism(v)) \mid \isomorphism(v) \in E'(\isomorphism(u))})  \\
& = (\algocolorrefinementattime{t}(G',\isomorphism(u)), \multiset{\algocolorrefinementattime{t}(G',v') \mid v' \in E'(\isomorphism(u))})  \\
& = \algocolorrefinementattime{t+1}(G',\isomorphism(u))
\end{align*}

Hence $\propP(t+1)$.

In particular:
\begin{itemize}
\item 
The condition of the repeat until loop is thus obtained for the same $t$ in $\algocolorrefinement(G)$ and $\algocolorrefinement(G')$.
\item $\histogramalgorefinement{\labelling_t} = \histogramalgorefinement{\labelling'_t}$.
\end{itemize}

So $\histogramalgorefinement G = \histogramalgorefinement {G'}$.
\end{proof}

The algorithm \algocolorrefinement does not characterize isomorphism, as shown in \Cref{figure:1wlcounterexample}.

\begin{proposition}\cite{immerman1990describing}
Let $G$ and $G'$ be two trees. 

$G, G'$ are isomorphic iff $G, G'$ are \algocolorrefinement-indistinguable.
\end{proposition}

\subsection{Analysis of failure}

\begin{proposition}
If two graphs with $n$ vertices with the same features are $d$-regular\footnote{A graph is $d$-regular if all vertices have the same degrees $d$.} then they are $\algocolorrefinement$-indistinguishable.
\end{proposition}

\begin{example}
\Cref{figure:1wlcounterexample} shows two non isomorphic 3-regular graphs with both 20 vertices. They are not isomorphic because decaprismane has 4-cycles while dodecahedrane does not.
\end{example}

Interestingly the probability that \algocolorrefinement fails on two graphs with $n$ vertices taken uniformly randomly goes to 0 when $n$ goes to $+\infty$.

\index{probability}
\begin{theorem} [\cite{DBLP:journals/siamcomp/BabaiES80}, \cite{DBLP:conf/fct/ArvindKRV15}]
Let $G_n, G'_n$ be two independent uniformly random graphs with $n$ vertices. We have:
$$\proba{\text{$G_n, G'_n$ are $\algocolorrefinement$-indistinguable} \mid \text{$G_n, G'_n$ are isomorphism}} 
\xrightarrow[n \rightarrow +\infty]{} 1.$$
\end{theorem}

\studentproject{Read \cite{DBLP:journals/siamcomp/BabaiES80}, \cite{DBLP:conf/fct/ArvindKRV15} and provide a proof of the above theorem}

\newcommand{\drawcycle}[3]{
    \foreach \i in {1,...,#2} {
        \node (#1\i) at ({360/#2 * (\i-1)}:#3) {C};
    }
    
    \foreach \i in {1,...,#2} {
        \pgfmathtruncatemacro{\j}{mod(\i,#2) + 1}
        \draw (#1\i) -- (#1\j);
    }

}

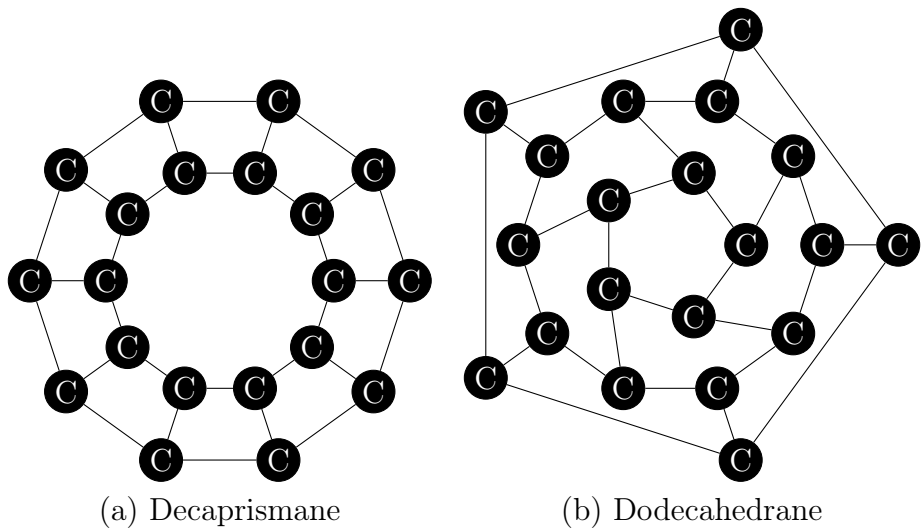
\begin{figure}
\centering

\begin{tabular}{cc}
\begin{tikzpicture}[every node/.style={circle,draw=black,text=white,fill=black,inner sep=1.5pt}]
    \def\n{10}
    \def\r{3} 
    
    \drawcycle {C} {10} {1.5}
    \drawcycle {D} {10} {2.5}
    
    \foreach \i in {1,...,10} {
    	\draw (C\i) -- (D\i);
    }
    
\end{tikzpicture}
&
\begin{tikzpicture}[every node/.style={circle,draw=black,text=white,fill=black,inner sep=1.5pt}]
    \def\n{10}
    \def\r{3} 
    
    \drawcycle {C} {10} {2}
    \drawcycle {D} {5} {1}
    \drawcycle {E} {5} {3}

    \foreach \i [evaluate=\i as \j using {int(\i*2)}] in {1,2,...,5} {
    	\draw (C\j) -- (D\i);
    }
    
    \foreach \i [evaluate=\i as \j using {int(\i*2-1)}] in {1,2,...,5} {
        	\draw (C\j) -- (E\i);
    }      
\end{tikzpicture}
\\
(a) Decaprismane & (b) Dodecahedrane
\end{tabular}
\index{decaprismane}
\index{dodecahedrane}
\caption{Decaprismane and dodecahedrane. Two non isomorphic graphs that are 1-WL-indistinguishable \cite{DBLP:journals/corr/abs-2003-04078}.
\label{figure:1wlcounterexample}}
\end{figure}

\section{Colour refinement and GNNs}

Intuitively, GNNs are weaker than \algocolorrefinement: \algocolorrefinement stores in the whole `colour' which correspond to all the arguments to compute the label of a vertex, while a GNN only stores the result.
This fact is stated in Theorem VIII.1 in \cite{DBLP:conf/lics/Grohe21}, as well as
\cite{DBLP:conf/iclr/XuHLJ19} \cite{DBLP:conf/aaai/0001RFHLRG19}.

\subsection{Colour refinement is self-contained}

Colour refinement contains all the information in order to compute the output of a GNN.

\begin{theorem}
Let $\aGNN$ be a GNN with $\numberoflayers$ layers. For all $t \in \set{0, \dots, \numberoflayers}$, for all pointed graphs $G, u$ and $G',u'$, we have:
\begin{center}
	$\algocolorrefinementattime{t}(G, u) = \algocolorrefinementattime{t}(G', u')$
	 then $\aGNNoutlayer{t}(G, u) = \aGNNoutlayer{t}(G', u')$.
	\end{center}
	
\end{theorem}

\begin{proof}
	We prove it induction on $t$.

\end{proof}

Said, differently, the partition of $\algocolorrefinementattime{\numberoflayers}(G)$ is finer than the one of $\aGNN(G)$. We make an abuse of notation and write $\aGNN(G)$ or $\aGNNoutlayer{t}(G)$ to denote the corresponding partition.

$$\algocolorrefinementattime{t}(G) \finerthan \aGNNoutlayer{t}(G).$$

In particular, we have $\algocolorrefinement(G) \finerthan \aGNNoutlayer{t}(G)$. So:

\begin{corollary}
$\algocolorrefinement(G) \finerthan \aGNNoutlayer{\numberoflayers}(G)$.
\end{corollary}

%

\subsection{But GNNs are powerful}

\index{colour refinement}
Is there a GNN that computes the $\algocolorrefinement$-partition?  
In Theorem VIII.4 of \cite{DBLP:conf/lics/Grohe21} partially answers the question. The answer is partial because the existence of a GNN is not uniform: we have one GNN for each size of graphs. We reformulate this result here.

\begin{theorem}
\label{theorem:aGNNforCR}
For all integers $n \in \setN$, there is a GNN $\aGNN$ such that for all graphs $G$ with at most $n$ vertices and where the initial labellings of each vertex is in $\set{0, 1}^k$, we have

 $$\aGNNoutlayer{n}(G) \finerthan \algocolorrefinement(G).$$

\end{theorem}

\studentproject{Read the proof of Theorem VIII.4 in \cite{DBLP:conf/lics/Grohe21}}

\section{Generalizations (*)}

\Cref{figure:1wlcounterexample} shows two non-isomorphic regular graphs that \algocolorrefinement is unable to distinguish. So the natural idea is to generalize 1-WL to tuples of vertices instead of single vertices.
We introduce the two main generalizations with the notations of \cite{DBLP:journals/jmlr/0001LMRKGFB23}: k-FWL = k-folkore WL, and k-OWL = k-oblivious WL.

\subsection{k-FWL test (for $k \geq 2$)}

\newcommand{\vectorsubstitutecoordinate}[3]{#1_{[#2 := #3]}}
\newcommand{\kfwlt}[2]{\textsf{fwl}_{#1}^{(#2)}}
\newcommand{\kowlt}[2]{\textsf{owl}_{#1}^{(#2)}}
\newcommand\twoowl[1]{\textsf{owl}_2^{(#1)}}

The algorithm $\algocolorrefinement$ has been extended on $k$-tuples of vertices. Given a labelled graph $G$, given a $k$-tuple $\vec v = (v_1, \dots, v_k)$ of vertices, we denote by $G[\vec v]$ the subgraph of $G$ induced by $\set{v_1, \dots, v_k}$. More precisely, $G[\vec v] = (V', E', \labellinginitial')$ is the graph whose vertices are $V' = \set{1, \dots, k}$ and $E' = \set{(i, j) \suchthat v_i E v_j}$ and $\labellinginitial'(i) := \labellinginitial(v_i)$.

\begin{itemize}
\item Initialization: the colour of $\vec v$ is $G[\vec v]$;
\item Refinement step: look at $k$-tuples overlapping with a tuple in all but one component. Make a \emph{single multiset} of \emph{tuples of values}.
\end{itemize}

Formally:
\begin{itemize}
\item $\kfwlt k 0 (G, \vec u) = G[\vec v]$;
\item $\kfwlt {k} {t+1}(G, \vec u) := (\kfwlt k t[\vec u], \multiset{
\left(\kfwlt k t(\vectorsubstitutecoordinate {G, \vec u} 1 w), 
\dots,
\kfwlt k t(\vectorsubstitutecoordinate {G, \vec u} k w)\right) \suchthat w \in V })$
\end{itemize}

where $\vectorsubstitutecoordinate {\vec u} i w$ 
is the vector $\vec u$ in which the $i$-th coordinate has been replaced by vertex $w$.

\subsection{k-OWL test (for $k \geq 2$)}

\begin{itemize}
\item Initialization: same as k-FWL;
\item Refinement step: look at $k$-tuples overlapping with a tuple in all but one component. Make \emph{several multisets} of \emph{values}.
\end{itemize}

\begin{itemize}
\item $\kowlt k 0 (G, \vec u) = G[\vec v]$;
\item 
$\kowlt k {t+1}[\vec u] := (\kowlt k t(\vec u), 
(\multiset{\kowlt k t(\vectorsubstitutecoordinate {\vec u} 1 w) \suchthat w \in V}, 
\dots,
\multiset{\kowlt k t(\vectorsubstitutecoordinate {\vec u} k w) \suchthat w \in V} ))$
\end{itemize}

\subsection{Relations}

\begin{proposition}
\algocolorrefinement and 2-OWL are as powerful.
\end{proposition}

\begin{proof}
2-OWL can be redefined as follows.

\begin{itemize}
\item $\twoowl0(u, v) = (\labellinginitial(u), \labellinginitial(v), 1_{uEv})$;
\item $\twoowl{t+1}(u, v) = (\twoowl{t}(u, v), \multiset{\twoowl{t}(w, v) \suchthat w \in V}, \multiset{\twoowl{t}(u, w) \suchthat w \in V} )$.
\end{itemize}

The goal is to prove that $\algocolorrefinementattime{t}(u) = \algocolorrefinementattime{t}(u')$ iff $\twoowl{t}(u, u) = \twoowl{t}(u', u')$.
We consider a stronger induction hypothesis which is denoted by $\propP(t)$: for all $u, v, u', v' \in V$, $\twoowl{t}(u, v) = \twoowl{t}(u', v')$ iff $\algocolorrefinementattime{t}(u) = \algocolorrefinementattime{t}(u')$ and $\algocolorrefinementattime{t}(v) = \algocolorrefinementattime{t}(v')$ and ($uEv$ iff $u'Ev'$).

\begin{itemize}
\item $\propP(0)$ is OK.
\item Suppose $\propP(t)$. Let us prove that $\propP(t+1)$.

\newcommand\textiff{\text{ iff }}
\newcommand\textand{\text{ and }}
\newcommand\textnot{\text{ not }}
\newcommand\textbijection{\text{ bijection }}
\newcommand{\explanation}[1]{\text{\footnotesize ~~~~~~~~(#1)}}
$\twoowl{t+1}(u, v) = \twoowl{t+1}(u', v')$
\begin{align*}
\textiff & \twoowl{t}(u, v) = \twoowl{t}(u', v') \\
 & \textand \multiset{\twoowl{t}(w, v) \suchthat w \in V} = \multiset{\twoowl{t}(w, v') \suchthat w \in V} \\
 & \textand \multiset{\twoowl{t}(u, w) \suchthat w \in V} = \multiset{\twoowl{t}(u', w) \suchthat w \in V} \\
  &  \explanation{By rewriting the multiset equalities with bijections} \\
 \textiff &  \algocolorrefinementattime{t}(u) = \algocolorrefinementattime{t}(u') \textand \algocolorrefinementattime{t}(v) = \algocolorrefinementattime{t}(v') \textand (uEv \textiff u'Ev') \\
   & \exists \textbijection \phi : V \rightarrow V \suchthat \forall w \in V, \twoowl{t}(w, v) = \twoowl{t}(\phi(w), v') \\
   & \exists \textbijection \psi : V \rightarrow V \suchthat \forall w \in V, \twoowl{t}(u, w) = \twoowl{t}(u', \psi(w)) \\
 %
  &  \explanation{By $\propP(t)$} \\
 \textiff &  \algocolorrefinementattime{t}(u) = \algocolorrefinementattime{t}(u') \textand \algocolorrefinementattime{t}(v) = \algocolorrefinementattime{t}(v') \textand (uEv \textiff u'Ev') \\
   & \exists \textbijection \phi : V \rightarrow V \suchthat \forall w \in V, \\
 &  \algocolorrefinementattime{t}(w) = \algocolorrefinementattime{t}(\phi(w)) \textand (wEv \textiff \phi(w)Ev') \\
   & \exists \textbijection \psi : V \rightarrow V \suchthat \forall w \in V, \\
 & \algocolorrefinementattime{t}(w) = \algocolorrefinementattime{t}(\psi(w)) \textand (u E w \textiff u' E \psi(w)) \\
  &  \explanation{By rewriting the bijection stuff with multiset equalities} \\
  \textiff &  \algocolorrefinementattime{t}(u) = \algocolorrefinementattime{t}(u') \textand \algocolorrefinementattime{t}(v) = \algocolorrefinementattime{t}(v')  \textand (uEv \textiff u'Ev') \\
& \multiset{\algocolorrefinementattime{t}(w) \suchthat w E v} = \multiset{\algocolorrefinementattime{t}(w') \suchthat w' E v'} \\
 &  \textand \multiset{\algocolorrefinementattime{t}(w) \suchthat \textnot w  E v} = \multiset{\algocolorrefinementattime{t}(w') \suchthat \textnot w'  E v'} \\
 & \multiset{\algocolorrefinementattime{t}(w) \suchthat u E w} = \multiset{\algocolorrefinementattime{t}(w') \suchthat u' E w'} \\
 &  \textand \multiset{\algocolorrefinementattime{t}(w) \suchthat \textnot u  E w} = \multiset{\algocolorrefinementattime{t}(w') \suchthat \textnot u'  E w'} \\
  &  \explanation{By removing the equalities with the not since the union is $\multiset{\algocolorrefinementattime{t}(w) \suchthat w \in W}$} \\
  \textiff &  \algocolorrefinementattime{t}(u) = \algocolorrefinementattime{t}(u') \textand \algocolorrefinementattime{t}(v) = \algocolorrefinementattime{t}(v')  \textand (uEv \textiff u'Ev') \\
 & \multiset{\algocolorrefinementattime{t}(w) \suchthat w E v} = \multiset{\algocolorrefinementattime{t}(w') \suchthat w' E v'} \\
 & \multiset{\algocolorrefinementattime{t}(w) \suchthat u E w} = \multiset{\algocolorrefinementattime{t}(w') \suchthat u' E w'} \\
%
 %
  &  \explanation{By definition of $\algocolorrefinementattime{t+1}$} \\
 \textiff & \algocolorrefinementattime{t+1}(u) = \algocolorrefinementattime{t+1}(u') \textand  \algocolorrefinementattime{t+1}(v) = \algocolorrefinementattime{t+1}(v')
\end{align*}

So $\propP(t+1)$.
\end{itemize}

\end{proof}

More generally:

\begin{proposition}
$k$-FWL is as powerful as $(k+1)$-OWL.
\end{proposition}

\begin{proposition}
$k+1$-OWL is strictly more powerful than $k$-OWL.
\end{proposition}

\studentproject{Prove the two last propositions}

\section*{Further reading}


\subsection{Higher-order GNNs}

Morris et al.
\cite{DBLP:conf/aaai/0001RFHLRG19}
have proposed generalization GNNs that corresponds to $k$-OWL for~$k > 2$.

\newpage
\section*{Exercises}

\begin{exercise}
Consider these two molecules (from \cite{DBLP:journals/corr/abs-2003-04078}):
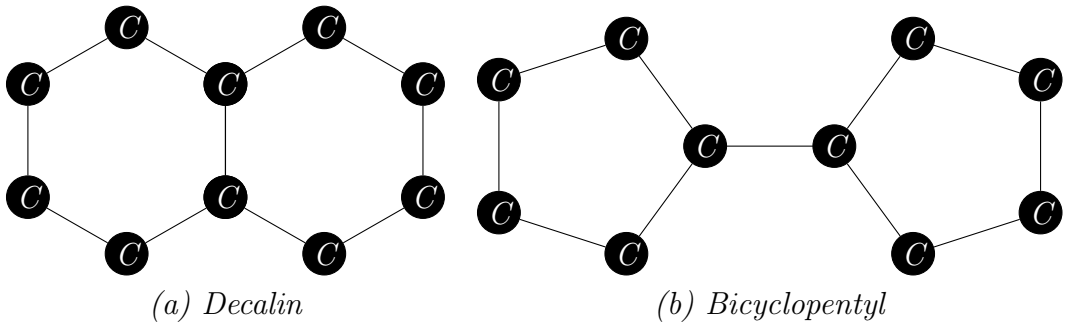
\begin{center}
\begin{tabular}{cc}
\begin{tikzpicture}
[every node/.style={circle,draw=black,text=white,fill=black,inner sep=1.5pt}]
\begin{scope}[rotate=30]
\drawcycle{C}{6}{1.5}
\end{scope}
\begin{scope}[xshift=26mm,rotate=-30]
\drawcycle{C}{6}{1.5}
\end{scope}
\end{tikzpicture}
& 
\begin{tikzpicture}
[every node/.style={circle,draw=black,text=white,fill=black,inner sep=1.5pt}]
\begin{scope}
\drawcycle{C}{5}{1.5}
\end{scope}
\begin{scope}[xshift=47mm,rotate=-36]
\drawcycle{D}{5}{1.5}
\end{scope}
\draw (1.5, 0) -- (3,0);
\end{tikzpicture} \\
(a) Decalin & (b) Bicyclopentyl
\end{tabular}
\end{center}

\begin{enumerate}
\item Are these two graphs isomorphic?
\item What is the output of color refinement?
\end{enumerate}
\end{exercise}

\begin{exercise}
Consider the two graphs $G$ and $H$ (Figure 1 in \cite{DBLP:conf/icassp/HuangV21}):

\begin{center}
\begin{tikzpicture}[scale=1, every node/.style={circle, draw, fill=white, inner sep=2pt}]

\node (a) at (0,0) {};
\node (b) at (1,0) {};
\node (c) at (0.5,1) {};

\draw (a) -- (b) -- (c) -- (a);

\node (d) at (2.5,0) {};
\node (e) at (3.5,0) {};
\node (f) at (3,1) {};

\draw (d) -- (e) -- (f) -- (d);

\node[draw=none, fill=none] at (1.5, -1) {\textit{G}};

\node (h1) at (6,0) {};
\node (h2) at (6.87,0.5) {};
\node (h3) at (6.87,1.5) {};
\node (h4) at (6,2) {};
\node (h5) at (5.13,1.5) {};
\node (h6) at (5.13,0.5) {};

\draw (h1) -- (h2) -- (h3) -- (h4) -- (h5) -- (h6) -- (h1);

\node[draw=none, fill=none] at (6, -1) {\textit{H}};

\end{tikzpicture}
\end{center}
\begin{enumerate}
\item Are $G$ and $H$ isomorphic?
\item Prove that color refinement does not distinguish $G$ and $H$.
\item Prove that $2-OWL$ does not distinguish $G$ and $H$.
\item Prove that $2-FWL$ does distinguish $G$ and $H$.
\end{enumerate}
\end{exercise}

\begin{exercise}
Play with \url{https://holgerdell.github.io/color-refinement/}
\end{exercise}

\begin{exercise}
Propose an efficient implementation of algorithm $\algocolorrefinement$.
\end{exercise}

\chapter{... and logic}

\fbox{
Key reference: \cite{DBLP:conf/iclr/BarceloKM0RS20}}

In this chapter we review the basics in logic, the complexity of the satisfiability problem and some connections with colour refinement and GNNs.

\section{First-order logic}
\index{first-order logic}

\newcommand\fotwo{FO\ensuremath{_2}}
\newcommand\foctwo{FOC\ensuremath{_2}}
\newcommand{\GML}{GML}

First-order logic (FO) validity/satisfiability problem is undecidable \cite{DBLP:journals/x/Turing37}. More precisely:

\begin{center}
\begin{tabular}{c|c|c}
& on finite models & on all models \\
\hline
satisfiability & undecidable in RE & undecidable in coRE \\
 & (enumerate the finite models) & \\
\hline
validity & undecidable in coRE & undecidable in RE \\
 & & (enumerate the finite proofs)
\end{tabular}
\end{center}

where RE means 'recursively enumerable'.

\section{Restricting the number of variables to two}

\index{variable}

 A noticeable decidable fragment is FO$_2$, the fragment of FO of formulas only containing two variables.

\begin{theorem} \cite{DBLP:journals/bsl/GradelKV97}
The satisfiability problem of FO$_2$ is NEXPTIME-complete.
\end{theorem}

\begin{proof}
\fbox{Upper bound}
Upper bound is obtained by small model property via \indexemph{Scott's formulas}. Details are given in \cite{DBLP:journals/bsl/GradelKV97} but we reproduce the proof for being self-contained.
Scott's reduction consists in defining $tr$ such that for all $\fotwo$-formulas $\phi$, $tr(\phi)$ is in the Gödel class ($\forall \forall \exists^*$-fragment), and $\phi$ and $tr(\phi)$ are equisatisfiable. To do we proceed as follows:
\begin{enumerate}
\item First we get rid of predicates of arity > 2 as follows.
\begin{enumerate}
\item  Consider an atomic subformula $R(v_1, \dots, v_n)$ where $v_1, \dots, v_n \in \set{x, y}$. 
\begin{itemize}
\item If $\set{v_1, \dots, v_n} = \set{x, y}$, replace $R(v_1, \dots, v_n)$ by $R^{(v_1, \dots, v_n)}(x, y)$ where $R^{(v_1, \dots, v_n)}$ is a fresh binary predicate
\item If $\set{v_1, \dots, v_n} = \set{x}$, replace $R(x, \dots, x)$ by $R^{(v_1, \dots, v_n)}(x)$ where $R^{(x, \dots, x)}$ is a fresh unary predicate
\item same for $y$
\end{itemize}
\item We finish by guaranteeing some equivalences.  For instance, if $R(x, y, x)$ and $R(y, x, y)$ both appears in $\phi$ we add:

$$\forall x \forall y (R^{(x, y, x)}(x, y) \lequiv R^{(y, x, y)}(y, x)).$$

etc.

\end{enumerate}

\item Now $\phi$ has only predicates of arity at most 2. We now perform a kind of \indexemph{Tseitin transformation}.

\newcommand\istrue[1]{\textsf{isTrue}_{#1}}
\begin{enumerate}
\item Each subformula $\psi$ is replaced by a predicate $\istrue{\psi}$ of arity 0, 1, 2 depending on the number of free variables in $\psi$
\item The final formula is $$tr(\phi) := \istrue{\phi} \land \lbigand_{\psi(\vec v) \text{ subformula}} \forall \vec v (\istrue \psi(\vec v) \lequiv meaning_\psi(\vec v))$$ where $meaning_\psi$ is:
\begin{enumerate}
\item $\psi$ if $\psi$ is atomic;
\item $\istrue\alpha(\vec v) \land \istrue \beta(\vec v)$ if $\psi = \alpha(\vec v) \land \beta(\vec v)$
\item $\lnot \istrue \alpha(\vec v)$ if $\psi = \lnot \alpha(\vec v)$
\item $\forall v~\istrue \alpha$ if $\psi = \forall v \alpha(\vec v)$
\end{enumerate}

Note that we get conjuncts with quantifiers $\forall \forall$ for (i-iii). For $(iv)$, because of the $\lequiv$ we get a $\forall \forall$ and a $\forall \exists$.
\end{enumerate}

\end{enumerate}

Now, we can group conjunct and rewrite $tr(\phi)$ as a formula of the form

$$ \forall x \forall y \alpha(x, y) \land \lbigand_{i=1..m} \forall x \exists y \beta_i(x, y)$$

where $\alpha$ and $\beta_i$ are quantifier-free formulas.

W.l.o.g. we can suppose that $\beta_i(x, y) \models (x \neq y)$.
 Indeed, if the model $\modelM$ has at least two elements we have that 
$$\modelM \models (\forall x \exists y \beta_i(x, y) \lequiv \forall x \exists y \underbrace{(x\neq y \land (\beta_i(x, x) \lor \beta_i(x, y))}_{\text{the new $\beta_i{x, y}$}}.$$

Now we define the notion of type. 
\begin{definition}
A \indexemph{type} $t(\vec v)$ is MCS over predicates and negation of predicates in $tr(\phi)$ involving only variables in $\vec v$. We say 1-type when $\vec v$ is variable and 2-type if $\vec v$ is $(x, y)$ or $(y, x)$.
\end{definition} 

More concretely, you can see a 1-type as a valuation on elements.

\begin{example}
Examples of 1-types are:
\begin{center}
$\set{red(x), tall(x), \lnot funny(x)}$, $\set{red(x), \lnot tall(x), funny(x)}$, etc.
\end{center} 
\end{example}

A 2-type is a valuation but over pairs of "elements" (variables because it is "abstract").

\begin{example}
An example of a 2-type is: 
$$\set{\begin{array}{l}
\lnot married(x, y),married(y, x), \lnot married(x, x), \lnot married(y, y), \\ fatherof(x, y), followsOnInstagram(x, y),  \dots
\end{array}}$$
\end{example}

Now, we consider a structure $\modelM$ that satisfies $tr(\phi)$. We will build a small model $\modelM'$ from it. Given an element $a$ in the domain of $\modelM$ the type $t_a$ of $a$ is the unique 1-type satisfied by~$a$.
Similarly, $t_{a,b}$ is the 2-type of the pair of elements $a, b$.

\begin{definition}
An element $a$ is a \indexemph{king} if $a$ is the only element in $\modelM$ to be of type $t_a$. 
\end{definition}

\newcommand{\skolemfunction}{\textsf{skolem}}
Let $K$ be the set of kings in $\modelM$.
Let $\skolemfunction_i$ be Skolem function for $\forall x \exists y \beta_i(x, y)$.

\begin{definition}
The \emph{royal court} is:
$$C := K \union \bigcup_i	\set{\skolemfunction_i(K)}$$

\begin{center}
\faCrown ~ \textcolor{red}{\faCrown} ~ $\skolemfunction_1(\text{\faCrown})$ ~ $\skolemfunction_2(\text{\faCrown})$ ~ $\skolemfunction_1(\text{\textcolor{red}\faCrown})$ ~ $\skolemfunction_2(\text{\textcolor{red}\faCrown})$ ~ $\dots$
\end{center}
\end{definition}

\newcommand{\domain}{D}

\newcommand{\othertypes}{\mathbb T}
Let $\othertypes$ be the 1-types appearing in $\modelM$ but not taken by a king:

$$\othertypes := \set{\text{1-types in $\modelM$}} \setminus \set{\text{1-types of some king in $\modelM$}}.$$

The domain of $\modelM'$ is:

$$\domain' := C \union D \union E \union F$$

where $D := \set{d_{it} \suchthat i=1..m, t \in \othertypes}$, $E := \set{e_{it} \suchthat i=1..m, t \in \othertypes}$, $F := \set{f_{it} \suchthat i=1..m, t \in \othertypes}$.

The interpretation of predicates are defined by the 1-types and 2-types we assign in $\modelM'$ to elements in $D'$ and pairs in $D'\times D'$.
\begin{enumerate}
\item \textbf{Definition of 1-types in $\modelM'$. } 
\begin{itemize}
\item The 1-type of an element of $C$ is its 1-type in $\modelM$.
\item The 1-type of $d_{it}, e_{it}, f_{it}$ is $t$.
\end{itemize}

\item \textbf{Definition of 2-types in $\modelM'$. } 
\begin{itemize}
\item 2-types for satisfying the $\beta_i$.

\begin{itemize}
\item For all kings $k \in K$, the 2-type of $(k,\skolemfunction_i(k))$ is just imported from $\modelM$.
\item For all $c \in C \setminus K$, if $\skolemfunction_i(c) \in C$, just import the 2-type of $(c,\skolemfunction_i(c))$ from $\modelM$.
If  $\skolemfunction_i(c) \not \in C$, the 2-type $(c,d_{i,t_c})$ in $\modelM'$ is $t_{c,\skolemfunction_i(c)}$.
\item Consider an element $d_{it} \in D$. Let $a$ be an element in $\modelM$ such that $t_a = t$. If $\skolemfunction_i(a)$ is a king, then $\skolemfunction_i(a) = \skolemfunction_i(d_{it})$. The 2-type of $(d_{it}, \skolemfunction_i(d_{it}))$ is $t_{a,\skolemfunction_i(a)}$. If $\skolemfunction_i(a)$ is not a king, $t_{\skolemfunction_i(a)} \in \othertypes$. The 2-type of $(d_{it}, e_{it_{\skolemfunction_i(a)}})$ is $t_{a,\skolemfunction_i(a)}$.
\item Same for $E,F$ in place of $D,E$. Same for $F,D$ in place of $E,F$.
\end{itemize}
\item Finally, for all elements $(e, e') \in D'^2$ for which we did not assign a 2-type yet, consider a pair $(a, a') \in D^2$ such that the 1-type of $e$ and $e'$ in $\modelM'$ are respectively $t_a$ and $t_{a'}$. We say that the 2-type of $(e, e')$ is $t_{a, a'}$.

\end{itemize}
\end{enumerate}

It remains to check that $\modelM' \models tr(\phi)$. This is left to the reader.

\newcommand{\settiles}{\mathbb T}
\newcommand{\formulatext}[1]{\colorbox{blue!10!white}{\text{#1}}}

\fbox{Lower bound}

Lower bound can be obtained from the NEXPTIME-hardness, see \cite{DBLP:conf/lam/Furer83} and \cite{DBLP:journals/bsl/GradelKV97}. We reproduce here the lower bound proof to avoid the reader to navigate throw the different papers.
To this aim, we give a reduction from the tiling problem with Wang tiles of a $2^n \times 2^n$-torus where $n$ is given in unary, and $\settiles$ is the set of tiles, and a give tile seed $t_0$.
We construct a $\fotwo$-formula $\phi$ as follows. A variables (eg. $x, y$) denotes a position of a cell in the torus.
We introduce unary predicates $X_i(x)$ for $i=0..n-1$ that says that the $i$-th bit of the $X$-coordinate is 1. Same for $Y_i(y)$ for the $Y$-coordinate. We introduce also $T_t(x)$ that says that tile $t$ is at $x$.
Before defining $\phi$ we introduce the macros:
\begin{itemize}
\item $Hsucc(x, y) := \formulatext{according to the $X_i(.)$ and $Y_i(.)$, $y$ is the next cell at right of $x$}$
\item $Vsucc(x, y) := \formulatext{according to the $X_i(.)$ and $Y_i(.)$, $y$ is the next cell at the top of $x$}$
\item $eq(x, y) := Heq(x, y) \land Veq(x, y)$
\end{itemize}

Now formula $\phi$ is the conjunction of the following formulas:

\begin{enumerate}
\item $\forall x \formulatext{there is a unique $t$ such that $T_t(x)$}$

\item same cell is really same cell: $\forall x \lbigand_{t \in \settiles} (T_t(x) \limply \forall y, eq(x, y) \limply T_t(y))$

\item $\exists x \formulatext{$x$ coordinate is $(0, 0)$}$

\item $\forall x \exists y Vsucc(x, y)$
\item $\forall x \exists y Hsucc(x, y)$

\item $\forall x \forall y (Hsucc(x, y) \limply \lbigor_{t, t' \text{horizontally compatible}} T_t(x)\land T_{t'}(y))$

\item same vertically

\end{enumerate}

\end{proof}

\begin{remark}
Five year later, Etessami et al. \cite{DBLP:journals/iandc/EtessamiVW02} have studed FO$_2$ on finite words and $\omega$-words. The corresponding satisfiability problem is also NEXPTIME-complete.
\end{remark}

More generally, $FO_k$ is the fragment of FO of formulas with at most $k$ variables.

\section{First-order logic with counting}

\index{counting}

We replace the quantification $\exists x \phi$ by $\exists^{\geq k}x \phi$ (there are at least $k$ elements $x$ such that $\phi$ holds). The obtained logic is called \emph{first-order logic with counting} (FOC).

\begin{definition}
We define $G, \lambda \models \exists^{\geq k}x \phi $ as follows:
\begin{align*}
\text{ there exists $u_1, \dots, u_k \in V$ all distinct such that for all $i = 1..k$, we have } G, \lambda[x := u_i] \models \phi.
\end{align*}
\end{definition}

We can thus define also the following macros for respectively "there are at most $k$ elements $x$ such that $\phi$ holds" and "there are exactly $k$ elements $x$ such that $\phi$ holds":
\begin{align*}
\exists^{\leq k} \phi & := \lnot \exists^{\geq k+1}x \phi \\
\exists^{= k} \phi & := \exists^{\geq k}x \phi \land \exists^{\leq k} \phi
\end{align*}

\begin{proposition}
FOC and FO (with equality) have the same expressivity.
\end{proposition}

\begin{proof}
$\exists^{\geq k}x \phi$ is rewritten in $$\exists x_1 \dots \exists x_k \left( 
\lbigand_{i<j} x_i \neq x_j ~~ \land ~~ \lbigand_i \phi(x_i)\right).$$
\end{proof}

However, FOC is interesting because it can lead to interesting fragments such as $FOC_2$ which is the two-variable fragment of FOC.

\begin{theorem}\cite{DBLP:conf/csl/Pratt-Hartmann14}
$FOC_2$ is NEXPTIME-complete.
\end{theorem}

\begin{proposition}
\cite{DBLP:journals/combinatorica/CaiFI92},
restated in \cite{DBLP:conf/iclr/BarceloKM0RS20})
We have:

$\algocolorrefinement(G, u) = \algocolorrefinement(G, v)$ iff for all $\phi(x) \in \foctwo$, ($G, u \models \phi$ iff $G, v \models \phi$).
\end{proposition}

The following proposition seems to be folklore.

\begin{proposition}
We have:

$\multiset{\algocolorrefinement(G)} = \multiset{\algocolorrefinement(G'}$ iff for all $\phi \in \foctwo$, ($G \models \phi$ iff $G' \models \phi$).
\end{proposition}

\section{Modal logic}

\index{modal logic}

\subsection{Syntax}
Modal logic extends \indexemph{propositional logic} with special operators $\lbox$ and $\ldiamond$ called \emph{modalities}\index{modality}. In the standard reading, the construction $\lbox \phi$ is read as $\phi$ is necessarily true.

\begin{definition}[language of basic modal logic]
A modal formula is a construction generated by the following rule:

$$\phi ::= \bottom \mid p \mid \lnot \phi \mid \phi \lor \phi \mid \ldiamond \phi$$

where $p$ ranges over the set of atomic propositions.
\end{definition}

\newcommand{\modaldepth}{md}
\begin{definition}
[modal depth]
Modal depth $\modaldepth(\phi)$ is defined by induction as follows:

\begin{align*}
\modaldepth (p) &  = p \\
\modaldepth (\lnot \phi) & = \modaldepth (\phi) \\
\modaldepth (\phi \lor \psi)&  = \max(\modaldepth(\phi), \modaldepth(\psi)) \\
\modaldepth(\ldiamond \phi) & = 1 + \modaldepth(\phi)
\end{align*}
\end{definition}

\subsection{Semantics}

Recall that a Kripke model is just a labelled graph.
%
%
%
A \emph{pointed Kripke model} is a pair $G, u$ where $G = (V, E, \labellinginitial)$ is a Kripke model and $v$ is a world in $V$.

\index{truth condition}
\begin{definition}[truth conditions]
Given $G = (V, E, \labellinginitial)$, $u \in V$, $\phi \in \logiclanguage{}$ we define $G, u \models \phi$ by structural induction over $\phi$:

\begin{itemize}
\item $G, u \not \models \bottom$;
\item $G, u  \models p$ iff $\labellinginitial(p) = 1$;
\item $G, u \models \lnot \phi$ iff $G, u \not \models \phi$;
\item $G, u \models \phi \lor \psi$ iff $G, u \models \phi$ or $G, u \models \psi$;
\item $G, u \models \ldiamond \phi$ iff there is a $v \in E(u)$ we have $G, v \models \phi$.
\end{itemize}
\end{definition}

We also introduce the dual modal construction $\lbox \phi$ which is equivalent to $\lnot \lbox \lnot \phi$.

\begin{example}
Construct a pointed graph satisfying $\lbox \ldiamond p$. Another one satisfying
$\ldiamond \lbox p$.
\end{example}

\subsection{Standard translation}

The standard translation consists in translating any modal formula $\phi$ into a first-order formula $\phi'(x)$ with one single free variable.

\begin{align*}
tr_x(p) & ::= p(x) \\
tr_x(\ldiamond \phi) & ::= \exists y xEy \land tr_y(\phi)
\end{align*}

It is interesting to note that ML is a fragment of \fotwo. 

\section{Graded Modal logic}
\index{graded modal logic}

\subsection{Definition}

Graded modal logic is like modal logic but operator $\ldiamond^{\geq k} \phi$. Its semantics is:

\begin{align*}
G, u \models \ldiamond^{\geq k} \phi & \text{ there are $k$ distinct $v_1, \dots, v_k$ such that for all $i = 1..k$ $u E v_i$ and $G, u_i \models \phi$}
\end{align*}

In the same way, GML is a fragment of $\foctwo$:

\begin{align*}
tr_x(p) & ::= p(x) \\
tr_x(\ldiamond^{\geq k} \phi) & ::= \exists^{\geq k} y xEy \land tr_y(\phi)
\end{align*}

At the end, we will know how to solve the satisfiability problem of GML. But let us start to tackle the satisfiability problem of K.

\subsection{Link with colour refinement}

\index{colour refinement}
\newcommand{\round}{t}

\begin{theorem}
\label{theorem:roundcolorGMLformulaexists}
For all rounds $\round$, 
for all colours $c$, there exists a GML formula $\phi_{\round, c}$ of modal depth $\round$ such that 
\begin{center}
$\algocolorrefinementattime{\round}(G, u) = c$ iff 
$G, u \models \phi_{\round, c}$.
\end{center}
\end{theorem}

\begin{proof}
By induction on $\round$.

\fbox{Base case. } We take $\phi_{0, c}$ to be a Boolean formula that describes $c$.

\begin{example}
If $c = \columnvector{2 \\ -6.5}$, then we take $\phi_{0, c} = (x_1 = 2) \land (x_2 = -6.5)$.
\end{example}

\fbox{Inductive case}
Consider the color $c = (c', M)$ where $c'$ is a color of round $r-1$, and $M$ is a multiset of colors of round $r-1$ too.
We set:

$$\phi_{\round, c} := \phi_{\round-1, c'} ~~ \land ~~ \lbigand_{c'' \in M} \ldiamond^{=count(c'', M)} \phi_{\round-1, c''}~~ \land ~~ \lbox \lbigor_{c'' \in M} \phi_{\round-1, c''}.$$

where $count(c'', M)$ is the number of occurrences of $c''$ in $M$.

\end{proof}

\begin{example}
The colour $(blue, \multiset{red, red, red, green, green, green, green, green})$ is captured by the formula
$$blue \land \ldiamond^{=3} red \land \ldiamond^{=5} green \land \lbox(red \lor green).$$
\end{example}

Now, we state that $(G, u)$ and $(G, u')$ are $\algocolorrefinement$-indistiguishable iff they satisfy the same formulas of $GML$.

\begin{theorem}[\cite{DBLP:conf/lics/Grohe21}, p. 6, Th. V.10] We have:
\begin{center}
$\algocolorrefinement(G, u) = \algocolorrefinement(G', u')$ iff for all $\phi \in GML$, $G,u\models \phi$ iff $G', u' \models \phi$.
\end{center}
\end{theorem}

\begin{proof}
(Proof given in Appendix in \cite{DBLP:conf/lics/Grohe21})

We prove the following property $\propP(\round)$ by induction on $\round$.
\begin{enumerate}
\item Color refinement gives the same results to $G, u$ and $G', u'$ after $\round$ rounds: $\algocolorrefinementattime{\round}(G, u) = \algocolorrefinementattime{\round}(G', u')$.
\item $G, u$ and $G', u'$ satisfy the same formulas $\phi$ in $GML$ of modal depth at most $\round$.
\end{enumerate}

\fbox{Base case. }

\begin{center}
Color refinement gives the same results to $G, u$ and $G', u'$ after $0$ rounds

iff 

$G, u$ and $G', u'$ have the same labellings

iff

$G, u$ and $G', u'$ satisfy the same Boolean formulas

iff

$G, u$ and $G', u'$ satisfy the same formulas $\phi$ in $GML$ of modal depth at most $0$.
\end{center}

\fbox{Inductive case} Suppose $\propP(\round-1)$ and prove $\propP(\round)$.
\begin{itemize}
\item ($1 \Rightarrow 2$) Suppose 1. Consider a GML-formula $\phi$. The formula $\phi$ is a Boolean combination of atoms $(x_i = 1)$ or subformulas $\ldiamond^{\geq N} \psi$. First, $G, u$ and $G', u'$ satisfy the same atoms.
By $\propP(\round-1)$, for all colours $c$, either all successors of $u$ coloured by $c$ all satisfy $\psi$ or none of them. As both $u$ and $u'$ have the same number of successors of a given color, we have $G, u \models \ldiamond^{\geq N} \psi$ iff $G', u' \models \ldiamond^{\geq N} \psi$. 
To conclude, $G, u \models \phi$ iff $G', u' \models \phi$.

\item ($2 \Rightarrow 1$) We prove not 1 $\Rightarrow$ not 2. Suppose that the colors of $u$ and $u'$ after $\round$ rounds are different. Then let $c$ be the colours of $u$ after $\round$ rounds. We have $G, u \models \phi_{t, c}$ while $G', u' \not \models \phi_{\round, c}$ where $\phi_{\round, c}$ is given by \Cref{theorem:roundcolorGMLformulaexists}. Hence not 2.
%
%
%
%
%
%
%
%

\end{itemize}

\end{proof}

\subsection{Link with GNNs}

\index{graded modal logic}
\index{graph neural network}

\begin{proposition}[Prop. 4.1 \cite{DBLP:conf/iclr/BarceloKM0RS20}]
For all \GML-formula $\phi$, there is a GNN $\aGNN$ such that for all $G,u$ we have:
\begin{center}
$G, u \models \phi$ iff $\aGNN(G,u) = \true$.
\end{center}
\end{proposition}

\begin{proof}
We postpone the proof to the next chapter, in which we prove a stronger result.
\end{proof}

\begin{proposition}[Prop. 4.2 \cite{DBLP:conf/iclr/BarceloKM0RS20}]
For all GNN $\aGNN$ that is FO-expressible, then there is a $\GML$-formula $\phi$ such that 

\begin{center}
$G, u \models \phi$ iff $\aGNN(G,u) = \true$.
\end{center}
\end{proposition}

It is sad to restrict ourselves to GNN that ar FO-expressible.

\section*{Exercises}

\begin{exercise}
Prove that there is a \foctwo-formula $\phi(x)$ for which there is no AC-GNN $\aGNN$ such that $G, u \models \phi(x)$ iff $\aGNN(G, u) = \true$.
\end{exercise}

\begin{exercise}
Prove that ML has the expressivity than GNNs where the aggregation function is MAX instead of SUM.
See \url{https://arxiv.org/abs/2507.18145}
\end{exercise}

\begin{exercise}
\newcommand{\gmlbisimulation}{\sim_{\#}}
We define the relation $\gmlbisimulation$ "graded bisimilation" defined in 
\url{https://www2.mathematik.tu-darmstadt.de/~otto/papers/cml19.pdf}.

Show that $G, u \gmlbisimulation G', u'$ iff for all $L$, the unvarallings up to $L$  of $G, u$ and $G', u'$ are isomorphic.
\end{exercise}

\begin{exercise}
In this exercise, we will prove that any GNN that is FO-definable is captured by a GML-formula.

\begin{enumerate}
\item Show that if for all $L$, the unravellings up to $L$, of $G, u$ and $G', u'$ are isomorphic, then for all GNNs $\aGNN$, we have $\aGNN(G, u) = \aGNN(G', u')$.
\item Read \url{https://www2.mathematik.tu-darmstadt.de/~otto/papers/cml19.pdf} that shows that the fragment of unary FO that only depend on the unravelling is GML.
\item Conclude.
\end{enumerate}
\end{exercise}


\chapter{Satisfiability problem of K}
\label{chapter:tableaumethod}

\newcommand{\tableaurulenondetonenode}[3]{\begin{center}
		\begin{tikzpicture}
		\tikzstyle{tableaunode} = [draw,outer sep=0,inner sep=2,minimum size=20]
		\node[tableaunode] (v1) at (0,0) {#1};
		\node[tableaunode, red] (v2) at (-3,-1) {\textcolor{red}{#2}};
		\node[tableaunode, red] (v3) at (3,-1) {\textcolor{red}{#3}};
		\draw[dashed, red] (v1) edge (v2);
		\draw[dashed, red] (v1) edge (v3);
		\end{tikzpicture}
\end{center}}

\newcommand{\tableaurulewithsuccs}[2]{\begin{center}
		\begin{tikzpicture}
		\tikzstyle{tableaunode} = [draw,outer sep=0,inner sep=2,minimum size=20]
		\node[tableaunode] (v2) at (4,0) {#1};
		\node[tableaunode] (v3) at (6,0) {\textcolor{red}{#2}};
		\draw[->]  (v2) edge (v3);
		\end{tikzpicture}
	\end{center}
}

\newcommand{\tableauruleonenode}[2]{\begin{center}
		\begin{tikzpicture}
		\tikzstyle{tableaunode} = [draw,outer sep=0,inner sep=2,minimum size=20]
		\node[tableaunode] (v1) at (-2,0) {\begin{tabular}{c}#1 \\ \textcolor{red}{#2}\end{tabular}};
		\end{tikzpicture}
	\end{center}
}

Before tackling the satisfiability problem for graded modal logic (and a richer logic called $\Ksharp$ introduced in \Cref{chapter:verificationGNNs}), it is good to concentrate on a simple setting.
In this chapter, we tackle the satisfiability problem for basic modal logic K:
\begin{itemize}
	\item input: a modal formula $\phi$;
	\item output: yes if $\phi$ is satisfiable; no otherwise.
\end{itemize}

We will explain the tableau method, an algorithm for deciding satisfiability problem.

\section{Negative normal form}

In the tableau method we will propose, we need disjunction, conjunction, box and diamonds are \emph{explicit}, e.g. no disjunction is hidden like in $\lnot (\phi \land \psi)$! That is why we introduce the notion of formula in \emph{negative normal form} where negations only appear in front of atomic propositions.

\begin{definition}[negative normal form]
A formula in \indexemph{negative normal form} belongs to the language defined by the rule

$$ \phi \grammaris p \grammarseparation \lnot p \grammarseparation \phi \lor \phi \grammarseparation \phi \land \phi \grammarseparation \ldiamond \phi \grammarseparation \lbox \phi$$

where $p$ ranges over the set of atomic propositions.
\end{definition}

We suppose that the formula $\phi$ (and all formulas) are in \emph{negative normal form}. If $\phi$ is not in negative normal form, apply these rewriting rules that pushes negations in front of atomic propositions:

\begin{center}
\begin{tabular}{lll}
$\lnot \lbox \phi$ & becomes & $\ldiamond \lnot \phi$  \\
$\lnot \ldiamond \phi$ & becomes & $\lbox \lnot \phi$  \\
$\lnot (\phi \land \psi)$ & becomes & $(\lnot \phi \lor \lnot \psi)$ \\
$\lnot (\phi \lor \psi)$ & becomes & $(\lnot \phi \land \lnot \psi)$
\end{tabular}
\end{center}

\index{tableau method}

\section{Overview}

In a nutshell, the tableau method is a proof system that constructs a Kripke structure. Each time a formula is written, it means that the formula \emph{should} hold at a given world. It can be seen as a procedure that rewrites a labelled graph.
The tableau method starts with an initial graph made up of one node containing~$\phi$:
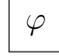
\begin{center}
	\begin{tikzpicture}
	\tikzstyle{tableaunode} = [draw,outer sep=0,inner sep=2,minimum size=20]
	\node[tableaunode] at (0,0) {$\phi$};
	\end{tikzpicture}
\end{center}

\section{Tableau rules}

Tableau rules there are rewriting rules that make explicit the meaning of `a formula \emph{should} hold'. Let us start with the rule for $\land$ (in red, we write what is added).

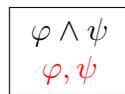
\tableauruleonenode{$\phi \land \psi$}{$\phi, \psi$}

The following rule for the $\lor$ connective is non-deterministic.
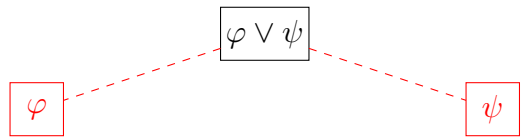
\tableaurulenondetonenode{$\phi \lor \psi$}{$\phi$}{$\psi$}

The clash rule for contradiction gives a clash.
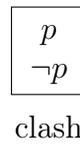
\begin{center}
\begin{tikzpicture}
\tikzstyle{tableaunode} = [draw,outer sep=0,inner sep=2,minimum size=20]
\node[tableaunode] (v1) at (0,0) {$\begin{array}{c}p\\ \lnot p\end{array}$};
\node (v2) at (0,-1) {clash};
\end{tikzpicture}
\end{center}

The rule for $\ldiamond$ adds a successor containing $\phi$ if there is no successor.

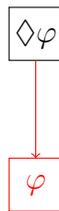
\begin{center}
\begin{tikzpicture}
\tikzstyle{tableaunode} = [draw,outer sep=0,inner sep=2,minimum size=20]
\node[tableaunode] (v2) at (4,0) {$\ldiamond \phi$};
\node[tableaunode, red] (v3) at (4,-2) {\textcolor{red}{$\phi$}};
\draw[->, red]  (v2) edge (v3);
\end{tikzpicture}
\end{center}

The rule for $\lbox$ adds $\phi$ in all successors.

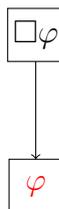
\begin{center}
\begin{tikzpicture}
\tikzstyle{tableaunode} = [draw,outer sep=0,inner sep=2,minimum size=20]
\node[tableaunode] (v2) at (4,0) {$\lbox \phi$};
\node[tableaunode] (v3) at (4,-2) {\textcolor{red}{$\phi$}};
\draw[->]  (v2) edge (v3);
\end{tikzpicture}
\end{center}

\newpage

\section{Example}

We draw dashed arrow for non-determinism and plain arrow for the relation in the model that is created.

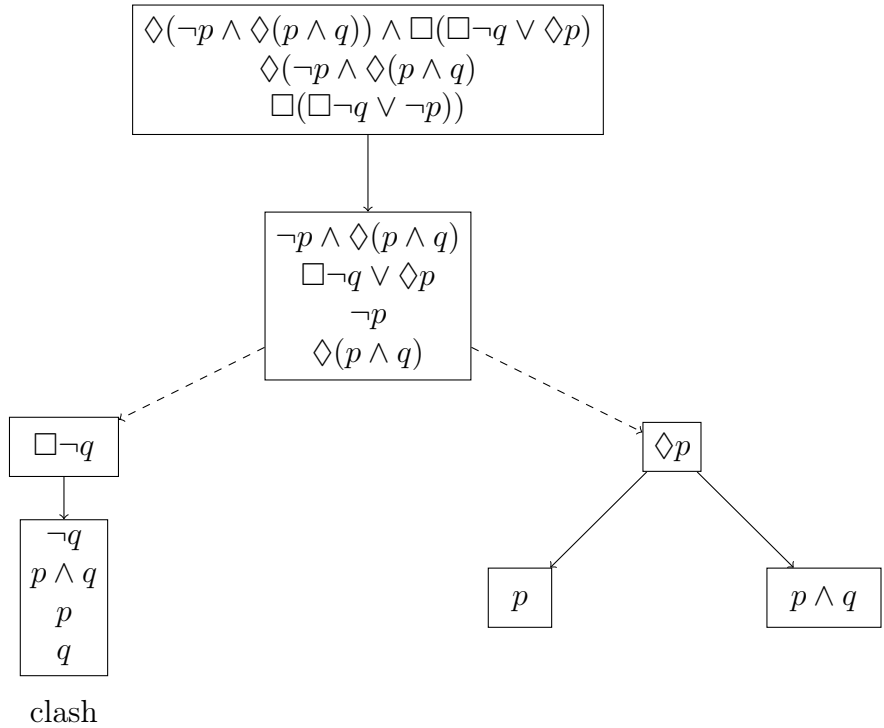
\begin{center}
\begin{tikzpicture}
\tikzstyle{prestate} = [draw, fill=gray!20];
\tikzstyle{state} = [draw, align=center];
\node[state, align=center] (v1) at (0,-1) {$\ldiamond(\lnot p \land\ldiamond (p \land q)) \land\lbox (\lbox \lnot q \lor \ldiamond p)$ \\ $\ldiamond(\lnot p \land \ldiamond(p \land q)$ \\ $\lbox (\lbox \lnot q \lor \lnot p))$ };
\node[state,  align=center] (v3) at (0, -4) {$\lnot p \land \ldiamond(p \land q)$ \\ $\lbox \lnot q \lor \ldiamond p$
\\
$\lnot p$ \\ $\ldiamond(p \land q)$};

\node[state] (v41) at (-4, -6) {$\begin{array}{c}\lbox \lnot q\end{array}$};

\node[state] (v41s) at (-4, -8) {$\lnot q$ \\
	$p \land q$ \\
	$p$ \\
	$q$};
\node[] (v41sflop) at (-4, -9.5) {clash};

\node[state] (v42) at (4, -6) {$\ldiamond p$};

\node[state] (v42s1) at (2, -8) {$\begin{array}{c}p \end{array}$};
\node[state] (v42s2) at (6, -8) {$\begin{array}{c}p \land q\end{array}$};

\draw[->] (v1) -- (v3);
\draw[->, dashed] (v3) -- (v41);
\draw[->, dashed] (v3) -- (v42);
\draw[->] (v41) -- (v41s);
\draw[->] (v42) -- (v42s1);
\draw[->] (v42) -- (v42s2);
\end{tikzpicture}
\end{center}

\section{Tableau system as a labelled proof system}

A tableau method works (and can be formalized) as a rewriting term system. At each time of the algorithm we maintain a set of terms of the following form:

\begin{itemize}
\item $(\sigma\;\phi)$, where $\sigma$ is a abstract symbol and $\phi$ is a formula. The term $(\sigma\;\phi)$ means that `$\phi$ should be true in the world denoted by $\sigma$'.
\item $(R\;\sigma\;\sigma')$ where $\sigma$ and $\sigma'$ are two symbols. The term $(R\;\sigma\;\sigma')$ means that `the world denoted by $\sigma$ is linked by relation $R$ to the world denoted by $\sigma'$.
\end{itemize}

\index{tableau rule}
\index{rule}

The tableau method starts with $(\sigma\;\phi)$ where $\sigma$ is a fresh symbol and $\phi$ is the formula we want to test. We then apply some rules.

\subsection{Boolean tableau rules}
\label{subsection:Booleantableaurules}
Let us start by defining the Boolean tableau rules for Boolean connectives:
$$\begin{array}{p{6cm}p{6cm}}\infer[\text{(Rule $\land$)}]{(\sigma\;\phi_1)(\sigma\;\phi_2)}{(\sigma\;\phi_1 \land \phi_2)}
& 
\infer[\text{(Rule $\lor$)}]{(\sigma\;\phi_1)\mid(\sigma\;\phi_2)}{(\sigma\;\phi_1 \lor \phi_2)} 
\\
\vspace{2mm}
\infer[\text{(Clash rule)}]{\text{rejecting execution}}{(\sigma\;p)(\sigma\;\neg p)}
\end{array}$$

Another presentation is to handle a set $\Gamma$ of formulas. The set $\Gamma$ corresponds to a given $\sigma$: it contains the formulas $\phi$ such that $(\sigma ~ \phi)$ is produced. The rules are then:
\begin{itemize}
\item If $\phi_1 \land \phi_2$ is in $\Gamma$, then add $\phi_1$ and $\phi_2$ to $\Gamma$;
\item If $\phi_1 \lor \phi_2$ is in $\Gamma$, then non-deterministically choose to either add $\phi_1$ in $\Gamma$, or add $\phi_2$ in $\Gamma$;
\item If $p$ and $\lnot p$ are $\Gamma$, then the execution is rejecting.
\end{itemize}

\subsection{Tableau rules for modal operators}

Now we define the rules for modal operators:
$$\begin{array}{p{6cm}p{6cm}}\infer[\text{(Rule $\ldiamond$})]{(R\;\sigma\;\sigmanew)(\sigmanew\;\phi)}{(\sigma\;\ldiamond\phi)}
& 
\infer[\text{(Rule $\lbox$)}]{(\sigma'\;\phi)}{(\sigma\;\lbox\phi)(R\;\sigma\;\sigma')}
\end{array}$$
where $\sigmanew$ is new fresh symbol.

\section{Soundness and completeness}

\begin{theorem}[soundness]
	If there is an execution of the tableau method that is not clashing, then the initial formula is satisfiable.
\end{theorem}

\begin{theorem}[completeness]
	If the initial formula is satisfiable then there is an execution of the tableau method that is not clashing.
\end{theorem}

We leave the proofs of these theorems to the next chapter.

\section*{Bibliographical notes}

	In \cite{DBLP:books/cu/BlackburnRV01} (p. 383), the definition of Hintikka set 
	[Blackburn p. 357, Def 6.24] is given. Their algorithm is cleaner in a sense, but not as easy to understand. 
	The tableau method presented here can be extended for graded modal logic \cite{DBLP:conf/cade/Tobies99}.
	
	%

Concerning implementation issues and optimisation, the reader may have a look to \cite{DBLP:books/el/07/HorrocksHSS07}.

\section*{Exercices}

\begin{enumerate}
	\item Apply the tableau method to the modal formula of your choice.
	\item Explain informally why any satisfiable modal formula is satisfiable in a tree. Can you bound its arity? its depth?
	\item How would you adapt the tableau method to know whether a given formula is true in a reflexive model ($u R u$ for all worlds $u$)?
\end{enumerate}

\chapter{Satisfiability problem of K is PSPACE-complete}
Savitch theorem says that PSPACE = NPSPACE. It means that proving a PSPACE upper bound can be proven by given a non-deterministic algorithm that requires a polynomial amount of space. 
We now turn the tableau method of Chapter~\ref{chapter:tableaumethod} we just saw into a non-deterministic algorithm. 

\index{PSPACE}

\section{Algorithm}

The design of our algorithm is as follows.
\begin{itemize}
	\item Boolean rules are applied in a given node are performed in the same call (see~\Cref{subsection:Booleantableaurules});
	\item Rules for modal operators are simulated by recursive calls. The number of recursive calls is thus bounded by the modal depth.
\end{itemize}
\begin{algo}
\algoprocedure satK($\Gamma$)

\begin{algobloc}

\algochoose outcomes of Boolean rules until $\Gamma$ is saturated

\algoif the clash rule can be applied on $\Gamma$ \algothen

\begin{algobloc}

 \algoreject

\end{algobloc}

\algofor all formulas of the form $\ldiamond \psi$ in $\Gamma$

\begin{algobloc}

satK($\psi \union \set{\chi \suchthat \lbox \chi \in \Gamma}$)

\end{algobloc}

\end{algobloc}

\end{algo}

\section{Soundness and completeness}

We denote the modal depth of $\phi$ by
$md(\phi)$. We define $md(\Gamma) = max_{\psi \in \Gamma} md(\psi)$.

\begin{theorem}[completeness]
If $\Gamma$ is satisfiable, then there exists an non-rejecting execution of $sat(\Gamma)$.
\end{theorem}

\begin{proof}
By induction on $md(\Gamma)$. Let $P(k)$ is

\begin{center}
`For all $\Gamma$ such that $md(\Gamma) \leq k$, if $\Gamma$ is satisfiable, then there exists a non-rejecting execution of $sat(\Gamma)$.'
\end{center}

\fbox{$md(\Gamma) = 0$}
There exists a model $\modelM = (W, R, V)$ and a world $w$ such that $\modelM, w \models \psi$ for all $\Gamma$. We prove the following invariant during one possible execution of the algorithm:
\begin{center}
for all $\psi \in \Gamma$, $\modelM, w \models \psi$.
\end{center}

Rule and: if $\psi_1 \land \psi_2 \in \Gamma$, then $\modelM, w \models \psi_1 \land \psi_2$. The rule and adds $\psi_1, \psi_2 \in \Gamma$. By the definition of the truth condition, $\modelM, w \models \psi_1$ and $\modelM, w \models \psi_2$.

Rule or: if $\psi_1 \lor \psi_2 \in \Gamma$, then $\modelM, w \models \psi_1 \lor \psi_2$. Either $\modelM, w \models \psi_1$ or $\modelM, w \models \psi_2$. Suppose that we are in the case where $\modelM, w \models \psi_2$. It is then sufficient to consider the execution where $\psi_2$ is added to $\Gamma$ and the invariant remains true.

As $\modelM, w  \models p$ and $\modelM, w \models \lnot p$ is impossible, the execution is non-rejecting.

\fbox{recursive case}

Suppose $P(k)$ and let us prove $P(k+1)$. Let $\Gamma$ such that $md(\Gamma) = k+1$. There exists a model $\modelM = (W, R, V)$ and a world $w$ such that $\modelM, w \models \psi$ for all $\Gamma$. The beginning of the proof is the same than for the basic case: we make the non-deterministic choices according to the truth of formulas in $w$.

Now, for all formulas of the form $\ldiamond \psi$ in $\Gamma$, as $\modelM, w \models \ldiamond \psi$ there exists $u \in R(w)$ such that $\modelM, u \models \ldiamond \psi$. More: we have $\modelM, u \models \chi$ for all $\lbox \chi \in \Gamma$.

Thus, $\psi \union \set{\chi \suchthat \lbox \chi \in \Gamma}$ is satisfiable. By $P(k)$, satK($\psi \union \set{\chi \suchthat \lbox \chi \in \Gamma}$) has an non-rejecting execution. We can thus construct an non-rejecting execution of satK$(\Gamma)$.
\end{proof}

\begin{theorem}[soundness]
If $sat(\Gamma)$ has an non-rejecting execution, then $\Gamma$ is satisfiable in a tree of depth $md(\Gamma)$ and of arity the number of $\ldiamond$ that appears in $\Gamma$.
\end{theorem}

\begin{proof}
$P(k)$ is defined as: 

\begin{center}
If $sat(\Gamma)$ has an non-rejecting execution, then $\Gamma$ is satisfiable in a tree of depth $md(\Gamma)$ and of arity the number of $\ldiamond$ that appears in $\Gamma$.
\end{center}

\fbox{basic case}

There is no modal operator. We define a model $\modelM$ made up of a unique world $w$ and the valuation $V(w) = \atmset \inter \Gamma$ when the saturation has been made.

We prove that by induction on $\phi \in \Gamma$ that $\modelM, w \models \phi$.

Propositions: ok

Negations of proposition: ok

Or: If $\phi \lor \psi \in \Gamma$, we have either $\phi$ or $\psi \in \Gamma$ because all the Boolean rules has been applied. Suppose it is $\psi \in \Gamma$. We have $\modelM, w \models \psi$. Thus, $\modelM, w \models \phi \lor \psi$.

And: idem.

\fbox{recursive case}
Suppose $P(k)$. Let us prove $P(k+1)$.

Let us take $\Gamma$ of model depth $k+1$ such that $sat(\Gamma)$ succeeds. Then satK($\psi \union \set{\chi \suchthat \lbox \chi \in \Gamma}$) succeeds for all $\ldiamond \psi \in \Gamma$ after Boolean saturation.

As shown in Figure~\ref{figure:gluing}, we construct $\modelM$ by gluing models obtained from the subcall. By induction, for all $\ldiamond \psi \in \Gamma$, we can find a tree $\modelM_{\psi}$ of depth at most $k$ and arity at most the number of $\ldiamond$ in $\psi \union \set{\chi \suchthat \lbox \chi \in \Gamma}$... well $\Gamma$.

We then create a root $w$ as in the basic case where $V(w) = \atmset \inter \Gamma$. We prove that by induction on $\phi \in \Gamma$ that $\modelM, w \models \phi$.

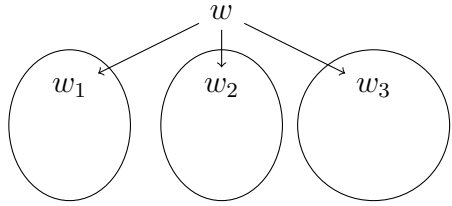
\begin{figure}
	\begin{center}
		\begin{tikzpicture}
		\node (w) {$w$};
		\node (1) at (-2, -1) {$w_1$};
		\node (2) at (-0, -1) {$w_2$};
		\node (3) at (2, -1) {$w_3$};
		\draw[->] (w) -- (1);
		\draw[->] (w) -- (2);
		\draw[->] (w) -- (3);
		\draw (-2, -1.5) ellipse (0.8cm and 1cm);
		\draw (0, -1.5) ellipse (0.8cm and 1cm);
		\draw (2, -1.5) ellipse (1cm and 1cm);
		
		\end{tikzpicture}
	\end{center}
	\caption{Model constructed by gluing models obtained from the subcall when diamond formulas are $\ldiamond \psi_1, \ldiamond\psi_2$ and $\ldiamond \psi_3$.\label{figure:gluing}}
\end{figure}

\end{proof}

\section{PSPACE upper bound}

\begin{theorem}
	The satisfiability problem for K is in PSPACE.
\end{theorem}

\begin{proof}
	The algorithm we saw is sound and complete. It runs in polynomial space and is non-deterministic. So the satisfiability problem for K is in NPSPACE. By Savitch's theorem, NPSPACE = PSPACE.
\end{proof}

\section{PSPACE lower bound}

\begin{theorem}
	The satisfiability problem for K is in PSPACE-hard.
\end{theorem}

\begin{proof}
By reduction from TQBF. Let us take a quantified binary formula $\exists p_1 \forall p_2 \dots \exists p_{2n-1} \forall p_{2n} \chi$ where $\chi$ is propositional.
The game behind TQBF can be represented by the binary tree in which $p_i$ is chosen at the $i$-th level. Player $\exists$ plays at the root level (level 0), at level 2, at level 4, etc. while player $\forall$ responds at level 1, at level 3, etc. The leaves correspond to all propositional valuations. Once we reach a leaf, the game ends and player $\exists$ wins iff the corresponding valuation satisfies $\chi$. 

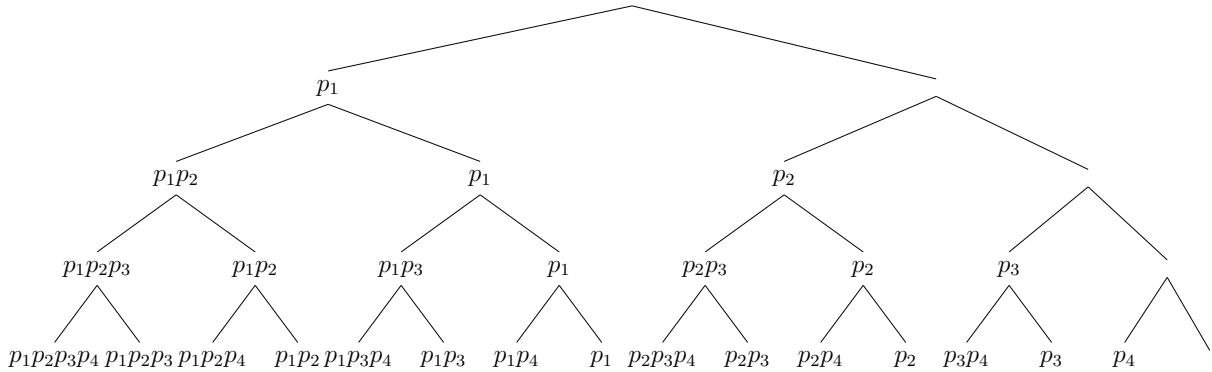
\begin{center}
\scalebox{0.8}{
	\begin{tikzpicture}[level distance=15mm,
	level 4/.style={sibling distance=14mm},
	level 3/.style={sibling distance=26mm},
	level 2/.style={sibling distance=50mm},
	level 1/.style={sibling distance=100mm},
	]
	\node {}
	child { node {$p_1$} 
		child { node {$p_1p_2$}
			child {node {$p_1p_2p_3$}
				child {node {$p_1p_2p_3p_4$}}
				child {node {$p_1p_2p_3$}}
			}
			child {node {$p_1p_2$}
				child {node {$p_1p_2p_4$}}
				child {node {$p_1p_2$}}
			}
		}
		child { node {$p_1$}
			child {node {$p_1p_3$}
				child {node {$p_1p_3p_4$}}
				child {node {$p_1p_3$}}
			}
			child {node {$p_1$}
				child {node {$p_1p_4$}}
				child {node {$p_1$}}
			}
		}
	}
	child { node {} 
		child { node {$p_2$}
			child {node {$p_2p_3$}
				child {node {$p_2p_3p_4$}}
				child {node {$p_2p_3$}}
			}
			child {node {$p_2$}
				child {node {$p_2p_4$}}
				child {node {$p_2$}}
			}
		}
		child { node {}
			child {node {$p_3$}
				child {node {$p_3p_4$}}
				child {node {$p_3$}}
			}
			child {node {}
				child {node {$p_4$}}
				child {node {}}
			}
		}
	};
	\end{tikzpicture}}
\end{center}

What is nice is that modal logic can express that a Kripke model contains the above tree:

$$TREE := \lbigand_{i=1}^{2n} \lbox ^{i-1} [\ldiamond p_i \land \ldiamond \lnot p_i \land \lbigand_{j<i} (p_j \lequiv \lbox p_j) \land (\lnot p_j \lequiv \lbox \lnot p_j)].$$

Formula $TREE$ explains explains how each $i$-th level works:

\begin{itemize}
\item  With $\lbox ^{i-1} = \lbox \dots \lbox$, we reach the $i-1$-th level;
\item $\ldiamond p_i \land^ \ldiamond \lnot p_i$ enforces the existence of a $p_i$-child and a $\lnot p_i$-child;
\item $\lbigand_{j<i} (p_j \lequiv \lbox p_j) \land (\lnot p_j \lequiv \lbox \lnot p_j)$ 
says the values of $p_1, \dots, p_{i-1}$ are copied to the successors.
\end{itemize}

On the input $\phi := \exists p_1 \forall p_2 \dots \exists p_{2n-1} \forall p_{2n} \chi$, the reduction computes in polynomial time the modal formula $tr(\phi) := TREE \land \ldiamond \lbox ... \ldiamond \lbox \chi$. The former $\phi$ is QBF-true iff the latter $tr(\phi)$ is K-satisfiable. Furthermore, $tr$ is computable in poly-time.
\end{proof}

\section*{Exercices}

\begin{enumerate}
	\item Write a deterministic algorithm for deciding the satisfiability for K.
	
	Hint: use backtracking.

	\item S5 is the modal logic interpreted in Kripke models where the relation is universal.
	\begin{enumerate}
		\item Prove that if $\phi$ is S5-satisfiable, then it is true in a model in which the number of worlds is bounded by the number of modal operators in $\phi$.
		\item Deduce that the satisfiability for S5 is in NP.
		\item Why is the satisfiability for S5 NP-hard?
	\end{enumerate}
	
	\item Prove that K is PSPACE-hard even with 0 variables! See \cite{DBLP:conf/aiml/ChagrovR02}
	
	\item Adapt the algorithm for the satisfiability for S4, the modal logic interpreted in Kripke models where the relation is reflexive and transitive. (difficult)
\end{enumerate}

\chapter{Verifying GNNs}
\label{chapter:verificationGNNs}

\fbox{
Key references: \cite{DBLP:conf/ijcai/NunnSST24} and \cite{DBLP:conf/icalp/BenediktLMT24}}

~

\index{verification}
Our goal is to be able to verify GNNs as follows.
We consider formulas in some logic evaluated on pointed graphs (for instance, modal logic or graded modal logic). Given a formula $\phi$, we write $\semone{\phi}$ for the set of pointed graphs $(G, u)$ satisfying $\phi$.
Let us list some verification tasks.
\begin{enumerate}
\item Satisfiability problem of a GNN: Given a GNN $\aGNN$, is there an input $(G, E, \labelling)$ that is positively classified by $\aGNN$? ($\semone{\aGNN} \neq \emptyset$)
\item Given a GNN $\aGNN$, given a specification formula $\phi$, are the inputs positively classified by $\aGNN$ exactly the inputs satisfying $\phi$? ($\semone{\aGNN} = \semone{\phi}$)
\item Given a GNN $\aGNN$, given a specification formula $\phi$, are the inputs positively classified by~$\aGNN$ satisfying $\phi$? ($\semone{\aGNN} \subseteq \semone{\phi}$)
\item Given a GNN $\aGNN$, given a specification formula $\phi$, are the inputs satisfying $\phi$ classified positively by $\aGNN$? ($\semone{\phi} \subseteq \semone{\aGNN}$)
\item Given a GNN $\aGNN$, given a specification formula $\phi$, does there exist an input satisfying $\phi$ and classified positively by $\aGNN$? ($\semone{\phi} \inter \semone{\aGNN} \neq \emptyset$)
\end{enumerate}

\cadregris{Link with Hoare logic}{
Problem 4 corresponds to the following Hoare logic triplet:

$$\set{\phi} N \set{output = true}$$

}

The methodology is to design a "superlogic" in which $\phi$ as well as the computation of the GNN $\aGNN$ can be described.

\section{Representing a GNN with a "logic"}

In this section, we introduce a \emph{lingua franca} to describe GNNs. The language just mimics the computation performed by a GNN. The language is inspired from the one in \cite{DBLP:journals/corr/abs-2502-16244}.

\newcommand{\gnnlogic}{GNN-logic\xspace}
\newcommand{\expression}{\vartheta} 
\newcommand{\agreggationfunction}{agg}
\newcommand{\setvertices}{V}
\newcommand\setedges E

\subsection{Syntax}

We consider expressions generated by the following grammar:

$$\expression ::= c \mid x_i \mid \activationfunction (\expression) \mid \agreggationfunction (\expression)
 \mid
\expression + \expression \mid c\times\expression$$

where $c$ is any number, $x_i$ is any feature, $\activationfunction$ is any symbol to denote any activation function, $\agreggationfunction$ is a symbol to represent any aggregation function (but it will interpreted as the sum in our case).

\subsection{Semantics}

The semantics mimics the computation performed by a GNN:

\begin{align*}
\sem {c}{G,u} & = c, \\
    \sem {x_i}{G,u} & = \ell(u)_i, \\
      \sem{\expression + \expression'}{G,u} & = \sem \expression{G,u} + \sem {\expression'}{G,u}, \\
      \sem{c \times \expression}{G,u} & = c \times \sem{\expression} {G,u}, \\
            \sem{\activationfunction (\expression)} {G,u} & = \semone{\activationfunction}(\sem \expression {G,u}),  \\
          \sem{\agreggationfunction (\expression)}{G,u} & = \Sigma_{v \mid u \setedges v}\sem {\expression}{G,v},
          \\
\end{align*}

We write $\sem{\expression \geq 1}{} = \set{G, u \mid \sem{\expression}{G, u} \geq 1}$.

\subsection{Correspondence with GNNs}

\begin{proposition}
Given a GNN $\aGNN$, there exists an expression $\expression$ such that\linebreak[4] $\semone{\aGNN}=\semone{\expression \geq 1}$.
\end{proposition}

\begin{proof}
See \cite{DBLP:journals/corr/abs-2502-16244} for a formal proof (even it is given for quantized GNNs). For being self-contained, we give a proof. Consider a GNN $N$:

\begin{itemize}
\item $\aGNNoutlayer{0}(G, u) = \labellinginitial(u)$;
\item $\aGNNoutlayer{t+1}(G,u) = \vec \activationfunction(\weightcolor{A_t} \times \aGNNoutlayer{t}(G,u) + \weightcolor{B_\timestep} \times \sum \multiset{\aGNNoutlayer{t}(G,v) \suchthat v \in E(u)} + \weightcolor{b_\timestep})$ for all $u \in V$.
\end{itemize}

with the stopping condition $\weightcolor{w}^t \labelling_\numberoflayers(u) + \weightcolor{b} \geq 0.$

The idea is to write the expression $\expression$.
To make it clearer we explain the proof via an example. Suppose:

\begin{itemize}
\item $\aGNNoutlayer{0}(G, u) = \labellinginitial(G, u)$;
\item $\aGNNoutlayer{1}(G,u) = \vec \activationfunction\left(\weightcolor{
\begin{pmatrix}
 2 & 1 \\
-1 & 4
\end{pmatrix}
} \times \aGNNoutlayer{0}(G,u) + 
\weightcolor{\begin{pmatrix}
 5 & 3 \\
2 & 6
\end{pmatrix}} \times \sum \multiset{\aGNNoutlayer{0}(G,v) \suchthat v \in E(u)} + \weightcolor{
\columnvector{1 \\ -2}
}\right)$ for all $u \in V$.
\item $\aGNNoutlayer{2}(G,u) = \vec \activationfunction\left(\weightcolor{
\begin{pmatrix}
 3 & 0 \\
-2 & 0
\end{pmatrix}
} \times \aGNNoutlayer{1}(G,u) + 
\weightcolor{\begin{pmatrix}
 -1 & 0 \\
0 & 5
\end{pmatrix}} \times \sum \multiset{\aGNNoutlayer{1}(G,v) \suchthat v \in E(u)} + \weightcolor{
\columnvector{0 \\ 0}
}\right)$ for all $u \in V$.
\end{itemize}

We suppose the stopping condition to $\weightcolor{2} \aGNNoutlayer{2}(G,v)[1] +\weightcolor{3}  \aGNNoutlayer{2}(G,v)[2] \geq 1$.

\end{proof}
\newcommand{\COMB}{comb}

Then, the corresponding \gnnlogic expression $\expression_\aGNN$ is given by: 

\begin{align*}
\psi_1 &= \alpha(2x_1 + x_2 + 5\agreggationfunction(x_1) - 3\agreggationfunction(x_2) + 1), \\ 
\psi_2 & := \alpha(-x_1 + 4x_2 + 2\agreggationfunction(x_1) + 6\agreggationfunction(x_2) - 2), \\
\chi_1 &:= \alpha(3\psi_1 - \agreggationfunction(\psi_1), \\
\chi_2 & := \alpha(-2\psi_1 + 5(\agreggationfunction(\psi_2)), \\
\varphi_A & := 2(\chi_1) - \chi_2 \geq 1.
\end{align*}

\begin{remark}
Note that some expressions do not represent a GNN. For instance, $\agreggationfunction(2\times x)$ does not syntactically correspond to a GNN. Indeed, aggregation is always computed on values of the form $\alpha(...)$.
\end{remark}

\section{Logic $K^\#$}

We now define logic $K^\#$ defined in \cite{DBLP:journals/corr/abs-2307-05150} and \cite{DBLP:conf/ijcai/NunnSST24}. A similar logic is defined in \cite{DBLP:conf/icalp/BenediktLMT24} but it does not have the $1_\phi$ construction. It has the same expressivity but probably not the same succinctness.

\subsection{Syntax}

\newcommand{\NTexpression}{\xi}
\newcommand{\Ap}{Ap}
\newcommand{\logicKsharpone}{K^\#}
\newcommand{\modalitynumber}{\#}
\newcommand{\istrue}[1]{1_{#1}}

Consider a countable set $\Ap$ of propositions. We define the language of logic $\logicKsharpone$ as the set of formulas generated by the following BNF:
\begin{align*}
	\phi & ::= p \mid \lnot \phi \mid \phi \lor \phi \mid \NTexpression \geq 0 \\ 
	\NTexpression & ::= c \mid \istrue\phi \mid \modalitynumber \phi \mid \NTexpression + \NTexpression \mid c\times \NTexpression 
\end{align*}

\subsection{Semantics}

\newcommand{\semanticsvalue}[2]{[[#1]]_{#2}}

As in modal logic, a formula $\phi$ is evaluated in a pointed graph $(G, u)$ (also known as pointed Kripke model). 
We define the truth conditions $(G,u) \models \phi$ ($\phi$ is true in $u$) by 
	\begin{center}
		\begin{tabular}{lll}
			$(G,u) \models p$ & if & $\labelling(u)(p) = 1$, \\
			$(G,u) \models \neg \phi$ & if & it is not the case that $(G,u) \models \phi$, \\
			$(G,u) \models \phi \land \psi$ & if & $(G,u) \models \phi$ and $(G,u) \models \psi$, \\
			$(G,u) \models \NTexpression \geq 0$ & if &  $\semanticsvalue{\NTexpression}{G,u} \geq 0$, \\
		\end{tabular}
	\end{center}
	and the semantics $\semanticsvalue{\NTexpression}{G,u}$ (the value of $\NTexpression$ in $u$) of an expression $\NTexpression$ by mutual induction on $\phi$ and $\NTexpression$ as follows.
	\begin{center}
		$\begin{array}{ll}
			\semanticsvalue{c}{G, u} & = c, \\
			\semanticsvalue{\NTexpression_1+\NTexpression_2}{G, u} & = \semanticsvalue{\NTexpression_1}{G,u}+\semanticsvalue{\NTexpression_2}{G,u}, \\
			\semanticsvalue{c \times \NTexpression}{G, u} & = c \times \semanticsvalue{\NTexpression}{G,u}, \\
			\semanticsvalue{\istrue\phi}{G, u} & = \begin{cases}
				1 & \text{if $(G,u) \models \phi$} \\
				0 & \text{else},
			\end{cases} \\  
			\semanticsvalue{\modalitynumber\phi}{G, u} & = \card{\{v \in \setvertices \mid (u,v) \in \setedges \text{ and } (G,v) \models \phi\}}.
		\end{array}$
	\end{center}
We illustrate it in the next example.
\begin{example}
	\begin{figure}
		\centering
		\begin{tikzpicture}[scale=2, rotate=90]
		\tikzstyle{vertex} = [circle, draw];
			\node[vertex] (u) at (-1, 0) {};
			\node[vertex] (v) at (-0.5, 1) {$p$};
			\node at (-0.5, 1.3) {$u$};
			\node[vertex] (w) at (0, 0) {$q$};
			\node[vertex] (y) at (-1, -1) {$p$};
			\draw[->] (v) edge (u);
			\draw[->] (v) edge (w);
			\draw[->] (u) edge (y);
		\end{tikzpicture}
		\caption{Example of a pointed graph $G, u$. We indicate true propositional variables at each vertex.}
		\label{fig:pointedgraph}
	\end{figure}
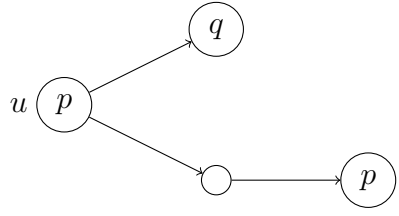
	
	Consider the pointed graph $G, u$ shown in Figure~\ref{fig:pointedgraph}. We have $G, u \models p \land (\modalitynumber \lnot p \geq 2) \land \modalitynumber (\modalitynumber p \geq 1) \leq 1$. Indeed, $p$ holds in $u$, $u$ has (at least) two successors in which $\lnot p$ holds. Moreover, there is (at most) one successor which has at least one $p$-successor.
\end{example}

\begin{proposition}
GML is a syntactic fragment of \Ksharp.
\end{proposition}

\begin{proof}
$\ldiamond^{\geq k} \phi$ is a shorthand for $\# \phi \geq k$.
\end{proof}

\section{$K^\#$ and truncReLU-GNNs}

\index{truncated ReLU}

\newcommand\truncReLU{\mathsf{truncReLU}}

In this section, we consider GNNs where the activation function $\activationfunction$ is $\truncReLU$. We also suppose that the features are in $\set{0, 1}$ (in \cite{DBLP:journals/corr/abs-2307-05150} and \cite{DBLP:conf/ijcai/NunnSST24}, graphs are tagged with propositions while in \cite{DBLP:conf/icalp/BenediktLMT24}, vertices have colors).

\subsection{From $K^\#$ to truncReLU-GNNs}

\newcommand{\mathgray}[1]
{\textcolor{gray}{\ensuremath{#1}}}

\newcommand{\trtoGNNL}{tr}

\begin{proposition}
\label{proposition:KsharptoGNN}
For all $K^\#$-formulas $\phi$, there is a GNN-expression $tr(\phi)$, computable in poly-time in $|\phi|$ such that $$\semone{\phi} = \semone{tr(\phi)\geq 1}.$$
\end{proposition}

\newcommand{\textiff}{\text{ iff }}

\begin{proof}

\begin{align*}
\mathgray{tr({\text{\Ksharp-expression or formula}})} & \mathgray{= \text{GNN-expression}} \\
tr(x_i=1)& =x_i \text{ provided $x_i$ takes its value in $\set{0, 1}$}\\
tr(\lnot \phi) & = 1 - \truncReLU(tr(\phi)) \\
tr(\phi \land \psi) & = \truncReLU(tr(\phi) + tr(\psi) - 1) \\
tr(\expression \geq 1) & = \truncReLU(\tau(\expression)) \\
\tau(\# \phi) &= \agreggationfunction(tr(\phi)) \\
\tau(\expression + \expression') & = \tau(\expression) + \tau(\expression') \\
\tau(1_\phi) & = tr(\phi) \\
\tau(c) & = c \\
\tau(c\expression) &= c\tau(\expression)
\end{align*}

Note that as constants $c$ are integers, the weights in the produced GNN $tr(\phi)$ are integers.
\begin{property}
For all $K^\#$-formulas $\phi$, 
$\sem{tr(\phi)}{G, u} \in \set{0, 1}$.
\end{property}
\begin{proof}
By definition of $tr$.
\end{proof}

We prove it by mutual induction on $\phi$ and on $\expression$ that:
\begin{enumerate}
\item $G, u \models \phi$ iff $\sem{tr(\phi)}{G, u} = 1$;
\item $\sem{\expression}{G, u} = \sem{\tau(\expression)}{G, u}$.
\end{enumerate}

\begin{itemize}
\item $G, u \models x_i = 1$ iff $\sem{x_i}{G, u} = 1$;
\item \begin{align*}
G, u \models \lnot \phi & \textiff G, u \not \models \phi \\
& \textiff \sem{tr(\phi)}{G, u} \neq 1 \\
&  \textiff  \sem{tr(\phi)}{G, u} =0 \\
&   \textiff  \sem{\truncReLU(tr(\phi))}{G, u} =0 \\
& \textiff \sem{tr(\lnot \phi)}{G, u} = 1.
\end{align*}

\item \begin{align*}
G, u \models \expression \geq 1 & \textiff \sem{\expression}{G, u} \geq 1 \\
 & \textiff \sem{\tau(\expression)}{G, u} \geq 1 \\
 & \textiff \sem{\truncReLU(\tau(\expression))}{G, u} = 1 \\
 & \textiff \sem{tr(\expression \geq 1)}{G, u} = 1
\end{align*}

\item \begin{align*}
\sem{\# \phi}{G, u} & = \card{\{v \in \setvertices \mid (u,v) \in \setedges \text{ and } (G,v) \models \phi\}} \\
 & = \card{\{v \in \setvertices \mid (u,v) \in \setedges \text{ and } \sem{tr(\phi)}{G, v} = 1\}} \\
 & = \sum_{v \mid u E v} \sem{tr(\phi)}{G, v} \\
 & = \sem{\agreggationfunction(tr(\phi))}{G, u}
\end{align*}
\end{itemize}

\end{proof}

\subsection{From truncReLU-GNNs to $K^\#$}
\newcommand{\trtoKsharp}{tr'}
\newcommand\trtoKsharpt{\trtoKsharp_t}
\newcommand\trtoKsharptm{\trtoKsharp_{t-1}}

The following translation function $\trtoKsharp$ takes a GNN described as an expression $\expression$ and gives an equivalent $\logicKsharpone$-expression.

\begin{proposition}
\label{proposition:truncrelugnnstoKsharp}
For all GNN-expressions $\expression$, there is a $\Ksharp$-expression such that \linebreak[4] $\semone{\expression \geq 1} = \semone{\trtoKsharp(\expression) \geq 1}$. Furthermore, $\trtoKsharp$ is computable in poly-time in $\expression$ and \textbf{ and the common denominator of weights (only good if it is represented in unary)}.
\end{proposition}

\begin{proof}

We first suppose that \textbf{weights are integers}, as in \cite{DBLP:journals/corr/abs-2307-05150} and \cite{DBLP:conf/ijcai/NunnSST24}.

\begin{align*}
\mathgray{\trtoKsharp({\text{expression for a GNN}})} & \mathgray{= \text{\Ksharp-expression}} \\
\trtoKsharp(x_i) &= x_i \\
\trtoKsharp(c) &= c \\
\trtoKsharp(c\expression)  &= c \times \trtoKsharp(\expression) \\
\trtoKsharp(\expression + \expression') &= \trtoKsharp(\expression) + \trtoKsharp(\expression') \\
\trtoKsharp(\truncReLU(\expression)) &= 1_{\trtoKsharp(\expression) \geq 1} \\
\trtoKsharp(\agreggationfunction(\expression)) &= \#(\trtoKsharp(\expression) \geq 1) \text{ provided $\expression$ is of the form $\truncReLU(.)$}
\end{align*}

\begin{proof}
We prove by induction on $\expression$ that $\sem{\expression}{G, u} = \sem{\trtoKsharp(\expression)}{G, u}$.
\end{proof}

\paragraph{Translating GNNs with rational weights. }
\newcommand{\truncReLUten}{\truncReLU_{M}}

So far we handled GNNs with integer weights.
In order to handle weights that are rationals, consider $M$ a common denominator. Let us analyse the computation of a GNN.

\begin{itemize}
\item $\aGNNoutlayer{0}(G, u) = \labellinginitial(u)$;
\item $\aGNNoutlayer{t}(G,u) = \vec \activationfunction(\weightcolor{A_t} \times \aGNNoutlayer{t-1}(G,u) + \weightcolor{B_\timestep} \times \sum \multiset{\aGNNoutlayer{t-1}(G,v) \suchthat v \in E(u)} + \weightcolor{b_\timestep})$ for all $u \in V$.
\end{itemize}

By induction on $t$, we prove that that $\aGNNoutlayer{t}(G,u) \in \set{0, \frac{1}{M^t}, \frac{2}{M^t}, \dots, 1}$. The idea is then to multiply by $M^t$ at round $t$ and to simulate the  
the activation function $\truncReLU$ by:

$$x \mapsto \max(0, \min(M^t, x)).$$

The translation is then:
\begin{align*}
\mathgray{\trtoKsharpt({\text{expression for a GNN}})} & \mathgray{= \text{\Ksharp-expression}} \\
\trtoKsharpt(x_i) &=  M^t \times x_i \\
\trtoKsharpt(c) &= M^t \times c \\
\trtoKsharpt(c\expression)  &=  M\times c \times \trtoKsharptm(\expression) \\
\trtoKsharpt(\expression + \expression') &= \trtoKsharpt(\expression) + \trtoKsharpt(\expression') \\
\trtoKsharpt(\truncReLU(\expression)) &=
1_{\trtoKsharpt(\expression) = 1} + 2 \times 1_{\trtoKsharpt(\expression) = 2} + \dots + M^t \times 1_{\trtoKsharpt(\expression) = M^t} \\ 
\trtoKsharpt(\agreggationfunction(\expression)) &= 
\#(\trtoKsharptm(\expression) = 1) + 2 \times \#(\trtoKsharptm(\expression) = 2) + \dots M^{t-1} \times \#(\trtoKsharptm(\expression) = M^{t-1})
\\
& ~~~~~~~~~~~~~~~~\text{ provided $\expression$ is of the form $\truncReLU(.)$}
\end{align*}

\begin{fact}
$\sem{\trtoKsharpt(\expression)}{G,u} = M^t \times \sem{\expression}{G,u}$.
\end{fact}

\begin{fact}
For all GNNL-expressions $\expression$ with $L$ layers, $\semone{\expression \geq 1} = \semone{\trtoKsharp_L(\expression) \geq M^L}$.
\end{fact}

\end{proof}

\begin{remark}
If $M$ is written in unary, there is no blow-up. Otherwise there is an exponential blow-up. The idea in the algorithm we will see at the end that it is sufficient to guess the value of $\trtoKsharp(\expression)$ between 0 and $M$. But there is still a blow-up problem for $\agreggationfunction$. Is it harmful? We will see later, later...
\end{remark}

In 
\cite{DBLP:conf/icalp/BenediktLMT24}, the authors explain a methodology to capture any activation functions that is piece-wise linear and eventually constant.

\section{Reduction to the satisfiability of $K^\#$}

We explain how to solve problem 4 ($\semone{\phi} \subseteq \semone{\aGNN}$). First we suppose that $\phi$ is already in $K^\#$.
 Then we use  \Cref{proposition:truncrelugnnstoKsharp} to get $\tau'(\aGNN)$.
 We then check that $\phi \rightarrow \tau'(\aGNN)$ is $\Ksharp$-valid, i.e. that $\lnot (\phi \rightarrow \tau'(\aGNN))$ is not $\Ksharp$-satisfiable.

\section{Going further}

\subsection{ReLU}
 \index{ReLU}
 \index{NEXPTIME}

\newcommand{\existsPAstar}{\exists PA^*}

Handling ReLU is more involved, and we do not have a clear correspondence with modal logic yet
\cite{DBLP:conf/icalp/BenediktLMT24}. 
In \cite{DBLP:conf/lics/HaaseZ19}, they generalize existential Presburger arithmetics with Kleene star as follows.

\begin{definition}
The Kleene star of a set $M \subseteq \setZ^d$ of vectors is defined by:

$$M^* := \bigcup_{k\geq 0}\set{\sum_{i=1}^k v_i \suchthat v_i \in M}.$$
\end{definition}

In other words, the Kleene star of $M$ contains any (finite) sum of vectors from $M$. 

We extend Presburger arithmetic with Kleene star.
 The obtained logic is \emph{Existential Presburger arithmetic with star} $\exists PA^*$. The syntax is:

$$ \phi, \psi, ... \grammaris \vec a \cdot \vec z \geq c \grammarsep \phi \land \psi \grammarsep \phi \lor \psi \grammarsep \exists y \phi \grammarsep \phi^*$$

The semantics $\semone\phi$ for $\phi$ is the set of vectors of values in $\setZ$ such that $\phi$ is true when we replace the free variables in $\phi$ by these values. By induction on $\phi$, we define $\semone\phi$ (w.l.o.g. we suppose that $\phi$ and $\psi$ contains the same free variables for $\phi \lor \psi$ and $\phi \land \psi$):

\begin{align*}
\semone{\vec a \cdot \vec z \geq c} & = \set{\vec x \in \setZ^n \suchthat \vec a \cdot \vec x \geq c} \\
\semone{\phi \lor \psi} & = \semone{\phi} \union \semone{\psi} \\
\semone{\phi \land \psi} & = \semone{\phi} \cap \semone{\psi} \\
\semone{\phi^*} & = \semone{\phi}^*
\end{align*}

\begin{theorem}(Th. III.1 in \cite{DBLP:conf/lics/HaaseZ19}) The satisfiability problem of $\existsPAstar$ is NEXPTIME-complete, and NP-complete if the number of nested star is bounded.
\end{theorem}

\begin{corollary}
(Th. 6.20 in \cite{DBLP:conf/icalp/BenediktLMT24})
The satisfiability problem of ReLU-GNNs is NEXPTIME-complete, and NP-complete if the number of layers is bounded.
\end{corollary}

\begin{proof}
\fbox{Lower bound}
 Lower bound (Th. 6.28 in \cite{DBLP:conf/icalp/BenediktLMT24}) is proven by reducing the  $\existsPAstar$-satisfiability to the satisfiability of a GNN with ReLU.
We reduce to the satisfiability of~$\existsPAstar$.

\fbox{Upper bound}
We only reproduce the proof of the upper bound here (Th. 6.26 in \cite{DBLP:conf/icalp/BenediktLMT24}). The Kleene-star is used to perform the aggregation over an unbounded number of successors.

Consider a GNN $\aGNN$ (w.l.o.g we suppose that the weights are integers, if not multiply by a common denominator).
We define formula $\phi_t(\vec x_0, \dots, \vec x_t)$ that says that $x_0, x_1, \dots, x_t$ is a possible sequence for $\aGNNoutlayer{0}(G,u), \dots, \aGNNoutlayer{t}(G,u)$ for some pointed graph $G, u$.

\begin{align}
\phi_0(\vec x_0) = & \lbigand_{i=1..d} (x_{0,i} = 0) \lor (x_{0,i} = 1) \\
\phi_t(\vec x_0, \dots ,\vec x_t) = & \phi_0(\vec x_0) \land \\
 & \exists \vec y_0, \dots \exists \vec y_{t-1} \\
 & \left( \phi_{t-1}(\vec y_0, \dots, \vec y_{t-1}) \right)^* \\
 & \land \exists \vec z_1, \dots \exists \vec z_{t} \\
 & ~~~~ \lbigand_{\tau=1..t} (\vec z_\tau = A_\tau \vec x_{\tau-1} + B_\tau \vec y_{\tau-1} + b_\tau)
 \\
 & ~~~~~~~~ \land \lbigand_i (z_{\tau,i} \leq 0)\land (x_{\tau,i}=0) \lor (z_{\tau,i} \geq 0)\land (x_{\tau,i}=z_{\tau,i}) 
\end{align}

(5.1) says that the initial value of each feature should be 0 or 1. The vector $\vec y_\tau$, see (5.3), is the one obtained by the aggregation, i.e. by summing on the outputs in the successors, see (5.4). Note that we grouped all the values so that the structure -- the successors of $G, u$ -- is guessed at once in the Kleene star (if we would have several Kleene stars, we would have several independent sets of "successors"). The vector $\vec z_\tau$ is the one obtained just after the combination (see 5.6), and before the activation function. In (5.7), we wrote the fact that $x_{\tau,i} = ReLU(z_{\tau, i})$. 

Our final formula is:

$$tr(\aGNN) := \exists \vec x_0 \dots \exists \vec x_L \phi_L(x_0, \dots, x_L) \land (w^t \times x_L + b \geq 0).$$

We have that there is a pointed graph $G, u$ such that $\aGNN(G,u) = true$ iff $tr(\aGNN)$ is $(PA)^*$-satisfiable.
\end{proof}

Our GNN-logic and $\existsPAstar$ are very similar, but they also differ. The semantics of a expression is a single real number, whereas the semantics of a $\existsPAstar$ formula is a list of values (one for each free variable). Having a single real number requires to evaluate wrt to a pointed graph $G, u$: we need to have the structure somewhere, to be sure that all $agg(....)$, the evaluation is according the same structure. In the semantics $\existsPAstar$ the structure is implicit. If we have several formulas $(..)*$, then the successors are different. That is why in \cite{DBLP:conf/icalp/BenediktLMT24} they pack all the relevant values in the same list, to have a single formula $(...)*$, i.e. to have the a single set of successors for a given vertex.

\subsection{Panorama}

We reproduce a (simplified version of the) table from \cite{DBLP:conf/icalp/BenediktLMT24} that sums up the situation for the satisfiability problem of a GNN.

\begin{center}
\begin{tabular}{|c|c|c|}
\hline
 & searching for directed graphs & searching for undirected graphs \\
\hline
truncated ReLU & PSPACE-complete & PSPACE-complete \\
\hline
ReLU & NEXPTIME-complete & undecidable \\
\hline
\end{tabular}
\end{center}

\subsection{Quantized GNNs}
 \index{quantization}

In real implementations, GNNs are quantized. For instance, numbers are 32-bit floats.
In \cite{DBLP:journals/corr/abs-2502-16244}, the authors prove that the satisfiability problem of \gnnlogic when numbers are quantized is in PSPACE-complete. The bitwidth $n$ is the number of bits used to represent the numbers.

\begin{theorem}
Verifying quantized GNNs where the bitwidth $n$ is given in unary is PSPACE-complete.
\end{theorem}

\section*{Open questions}

\begin{itemize}
\item Nice logic for truncReLU without the exponential blow-up. In \cite{DBLP:conf/icalp/BenediktLMT24}, they seem to treat the problem algorithmically but not with a poly-time translation in a logic
\item Parametrized complexity of $\existsPAstar$ wrt to the nested number of stars?
\item Design a "neat" modal logic that is equivalent to GNN with ReLU?
\item Parametrized complexity of verifying quantized GNNs wrt to the bitwidth
\item Lower bounds for other activation functions than truncReLU for quantized GNNs
\item Proof systems for the introduced logics
\end{itemize}

\newpage
\section*{Exercises}

\begin{exercise}
Write formally the translation from GML into $K^\#$ preserving the semantics.
\end{exercise}

\begin{exercise}
Prove that $K^\#$ is more expressive than FO.

\hfill Hint: see Appendix in  \cite{DBLP:conf/ijcai/NunnSST24}
\end{exercise}

\begin{exercise}
Prove that we can reduce in poly-time the satisfiability problem of $\Ksharp$ to the satisfiability problem of a $\Ksharp$-formula without any occurrence of $1_\psi$.

\hfill Hint: see  \cite{DBLP:conf/ijcai/NunnSST24}
\end{exercise}

\chapter{Satisfiability problem of logic with counting}
\label{chapter:qfbapa}

In this chapter, we focus on the satisfiability problem of a $K^\#$-formula.

\section{Difficulty to get a PSPACE Tableau Method for $K^\#$}

As explained in
\cite{DBLP:journals/corr/abs-2307-05150}, it is not sufficient, as for K, to prove consistency in successors. We also have to take into account implicit counting relations. For instance, we always have:

$$\# p~~ +~~ \# \lnot p ~~~=~~~~ \# q ~~+~~ \# \lnot q.$$

To take these constraints into account, we rely on QFBAPA (Quantiﬁer-free Fragment Boolean Algebra with Presburger Arithmetic) which combines Presburger arithmetic (reasoning about linear inequalities) and Boolean algebra (reasoning about $p$, $\lnot p$, $q$, $\lnot q$, and all Boolean formulas).

Despite we will not use the full power of QFBAPA, it is interesting to present QFBAPA on its own.

\section{Quantiﬁer-free Fragment Boolean Algebra with Presburger Arithmetic}
 \index{quantiﬁer-free Fragment Boolean Algebra with Presburger Arithmetic}

\newcommand{\setvariable}{S}
\newcommand{\setexpression}{B}
\newcommand{\universe}{\mathcal U}
\newcommand{\integerexpression}{E}

A QFBAPA formula is propositional formula where each atom is either an inclusion of sets or equality of sets or linear constraints 
\cite{kuncak-rinard-QFBAPA}.
 Sets are denoted by Boolean algebra expression, e.g., $(\setvariable \cup \setvariable') \setminus \setvariable''$, or $\universe$ where $\universe$ denotes the set of all points in some domain. Here $\setvariable$, $\setvariable'$, etc. are set variables.
Linear constraints are over $|\setexpression|$ denoting the cardinality of the set denoted by the set expression $\setexpression$. Let us give a formal definition of the syntax and semantics.

\newcommand{\integervariable}{x}

\begin{definition}(see Figure 1 in \cite{kuncak-rinard-QFBAPA})
A QBFBAPA formula is generated by the axiom~$\phi$ in the following BNF grammar:
\begin{align*}
\phi & ::= A \mid \phi_1 \lor \phi_2 \mid \lnot \phi \\
A & ::= B_1 = B_2 \mid B_1 \subseteq B_2 \mid E_1 = E_2 \mid E_1 \leq E_2 \\
B & ::= \setvariable \mid \emptyset \mid \universe \mid B_1 \union B_2 \mid \compl{B} \\
E & ::= \integervariable \mid k \mid E_1 + E_2 \mid k \times E \mid |B|
\end{align*}
where $\setvariable$ ranges in a countable set of set variables, $\integervariable$ ranges in a countable set of integer variables, $k$ ranges in $\setZ$.
\end{definition}

\begin{example}[of a QFBAPA formula]
$|pianist \cap student| + x \geq 5 ~~ \land~~(|pianist| \leq 10 \lor |student| \leq 10)$ 
\end{example}

The original QFBAPA \cite{kuncak-rinard-QFBAPA} also contains the construction $k~\textsf{divides}~E$ where $k$ is an integer and $E$ an expression. We omit it here since we do not use it. They also have a constant $\textsf{MAXC}$ which is always equal to $|\universe|$.

\newcommand\domain D
\newcommand\interpreted{}

\begin{definition}
A QBFBAPA model is a tuple $\modelM = (\domain, \semone{\cdots})$ where
\begin{itemize}
\item $\domain$ is a (possible empty) finite set, called the \indexemph{domain};
\item for all integer variables $\integervariable$, $\semone{\integervariable} \in \ensZ$;
\item for all set variables $\setvariable$, $\semone{\setvariable} \subseteq \domain$.
\end{itemize}
\end{definition}

We naturally extends $\semone{\cdots}$ to integer expressions $\integerexpression$ and set expressions $\setexpression$ as follows:
\begin{align*}
\semone{k} & := k \\
\semone{E_1 + E_2} & := \semone{E_1} + \semone{E_2} \\
\semone{k \times E} & := k \times \semone{E} \\
\semone{|B|} & := |\semone{B}|\\
\semone{B_1 \union B_2} & := \semone{B_1} \union \semone{B_2} \\
\semone{\compl{B}} & := \compl{\semone{B}} \\
\semone{\universe} & := \domain
\end{align*}

\begin{definition}
The truth conditions are given as follows:
\begin{align*}
\modelM \models B_1 = B_2 & \textiff \semone{B_1} = \semone{B_2} \\
\modelM \models B_1 \subseteq B_2 & \textiff \semone{B_1} \subseteq \semone{B_2} \\
\modelM \models E_1 = E_2 & \textiff \semone{E_1} = \semone{E_2} \\
\modelM \models E_1 \leq E_2 & \textiff \semone{E_1} \leq \semone{E_2}
\end{align*}

\end{definition}

We are going to prove:

\index{NP}
\begin{theorem}
QFBAPA satisfiability problem is in NP.
\end{theorem}

To prove that we rely on the fact that QFPA (quantifier-free Presburger arithmetics) is in NP, just a non-deterministic step (to resolve disjunction) on integer linear programming which is in NP \cite{DBLP:journals/jacm/Papadimitriou81}. 


\section{Fail for proving NP by naïve argument}

We first discuss the fact that knowing the size (written in binary) does not help much.
Consider the following formula.

$$	|\mathcal{U}| = n \ \land \ \bigwedge_{0 \leq i < j \leq m} |\setvariable_i \cup \setvariable_j| = 30 \ \land \ \bigwedge_{0 \leq i \leq m} |\setvariable_i| = 20  $$

A certificate would consist in telling for each set variable $\setvariable_i$ which elements are in $\semone{\setvariable_i}$. So each set variable is represented by a word $\set{0, 1}^{n}$ while $n$ is given in binary:

\begin{center}
\begin{tabular}{|l|l|}
\hline
& certificate (in/out for each elements in the domain $\domain$) \\
\hline
$S_1$ & 010010101010 \\
$S_2$ & 000010001010 \\
\vdots & \vdots \\
$S_m$ & 111011111100 \\
\hline
\end{tabular}
\end{center}

 The certificate is of exponential size in the size of the formula. So it seems that the satisfiability problem of QFBAPA is in NEXPTIME.

\section{Venn diagrams}

\newcommand{\coderegion}{\rho}
\newcommand{\sizeofregion}[1]{s_{#1}}
\newcommand{\nbbooleanexpressions}{d}
\newcommand{\nbvariables}{e}

\index{Venn diagram}
A Venn diagram is a picture that contains all the possible intersections called \indexemph{regions}, see \Cref{figure:venndiagram}. Each region is denoted by a \emph{region code} $\coderegion \in \set{0, 1}^\nbvariables$ telling whether the set $S_i$ is in or out. For instance, the code of $\compl{\setvariable_1} \inter \compl{\setvariable_2} \inter \setvariable_3$ is $001$ (out of $S_1$, out of $S_2$, in $S_3$).

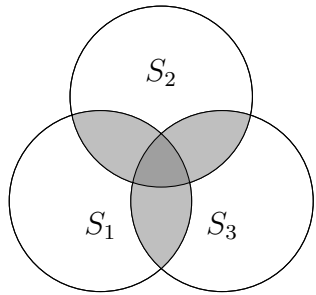
\begin{figure}[h]
\begin{center}
 \begin{tikzpicture}[scale=0.8]
        \def\firstcircle{(0,0) circle (1.5cm)}
        \def\secondcircle{(60:2cm) circle (1.5cm)}
        \def\thirdcircle{(0:2cm) circle (1.5cm)}

        \draw \firstcircle node[below] {$\setvariable_1$};
        \draw \secondcircle node [above] {$\setvariable_2$};
        \draw \thirdcircle node [below] {$\setvariable_3$};

        \begin{scope}
            \clip \firstcircle;
            \fill[gray!50] \secondcircle;
            \fill[gray!50] \thirdcircle;
        \end{scope}

        \begin{scope}
            \clip \secondcircle;
            \fill[gray!50] \thirdcircle;
        \end{scope}

        \begin{scope}
            \clip \firstcircle;
            \clip \secondcircle;
            \fill[gray!75] \thirdcircle;
        \end{scope}

        \draw \firstcircle \secondcircle \thirdcircle;


    \end{tikzpicture}
\end{center}
\caption{A Venn diagram with 3 set variables is made up of 8 regions.\label{figure:venndiagram}}
\end{figure}

The idea is to reason about the size of each region obtained by intersection and introducing a variable $\sizeofregion{01010110}$ to denote that size:

\begin{center}
$\begin{array}{l|l}
\text{Regions} & \text{Size of that region} \\
\hline
\compl{\setvariable_1} \inter \compl{\setvariable_2} \inter \cdots \inter \compl{\setvariable_\nbvariables} & \sizeofregion{000\ldots0} \\
\hline
\compl{\setvariable_1} \inter \compl{\setvariable_2} \inter \cdots \inter \setvariable_\nbvariables & \sizeofregion{000\ldots1} \\
\hline
\vdots & \vdots \\
\hline
\setvariable_1 \inter \setvariable_2 \inter \cdots \inter \compl{\setvariable_\nbvariables} & \sizeofregion{111\ldots0} \\
\hline
\setvariable_1 \inter \setvariable_2 \inter \cdots \inter \setvariable_\nbvariables & \sizeofregion{111\ldots1}
\end{array}$ 
\end{center}

\begin{example}
For instance for $n=3$ we have:

\begin{center}
$\begin{array}{l|l}
\text{Regions} & \text{Size of that region} \\
\hline
\compl{\setvariable_1} \inter \compl{\setvariable_2} \inter \compl{\setvariable_3} & \sizeofregion{000} \\
\hline
\compl{\setvariable_1} \inter \compl{\setvariable_2} \inter \setvariable_3 & \sizeofregion{001} \\
\hline
\compl{\setvariable_1} \inter \setvariable_2 \inter \compl{\setvariable_3} & \sizeofregion{010} \\
\hline
\compl{\setvariable_1} \inter \setvariable_2 \inter \setvariable_3 & \sizeofregion{011} \\
\hline
\setvariable_1 \inter \compl{\setvariable_2} \inter \compl{\setvariable_3} & \sizeofregion{100} \\
\hline
\setvariable_1 \inter \compl{\setvariable_2} \inter \setvariable_3 & \sizeofregion{101} \\
\hline
\setvariable_1 \inter \setvariable_2 \inter \compl{\setvariable_3} & \sizeofregion{110} \\
\hline
\setvariable_1 \inter \setvariable_2 \inter \setvariable_3 & \sizeofregion{111}
\end{array}$
\end{center}
\end{example}

We could rewrite any QFBAPA-formula into a QFPA-formula using variables $\sizeofregion{0101010111}$ as follows. We replace each cardinality expression with sums of the appropriate variables $\sizeofregion{0101010111}$. For instance:

\begin{align*}
|S_1 \inter S_2| & = \sizeofregion{110} + \sizeofregion{111} \\ 
|S_1 \inter S_2 \inter \compl{S_3}| & = \sizeofregion{110} \\
|S_1| & = \sizeofregion{100} + \sizeofregion{101} + \sizeofregion{110} + \sizeofregion{111}  \\
|\compl{S_1}| & = \sizeofregion{000} + \sizeofregion{001} + \sizeofregion{010} + \sizeofregion{011}  
\end{align*}

But there are an exponential number of variables in the number of set variables that will be used. This idea is good, but implemented naïvely, we get that the satisfiability problem of QFBAPA is in NEXPTIME.
Later on we will see that a poly-number of non-empty regions is sufficient. From that, we get NP membership. So let us go into this naïve reduction first.

\section{Naïve reduction to QFPA}

We first explain the naïve reduction to quantifier-free Presburger arithmetics. It is made of several steps of rewriting.

\paragraph{Getting rid of inclusions and set equality.}
We perform the following rewritings:

\begin{center}
\begin{tabular}{ccc}
 $\setexpression = \setexpression'$ & rewriting into & $\setexpression \subseteq \setexpression' \land \setexpression' \subseteq \setexpression$ \\
 $\setexpression \subseteq \setexpression'$ & rewriting into & $|\setexpression \inter \compl{\setexpression'}| = 0$.
\end{tabular}
\end{center}

\paragraph{Variables for cardinalities of Boolean expressions.}

\newcommand{\sizeofsetvariable}[1]{k_{#1}}
\newcommand{\numberofsetexpressions}{d}
Let $\setexpression_1, \dots, \setexpression_\numberofsetexpressions$ be the set expressions in $\phi$.
Instead of writing $|\setexpression_i|$, we introduce a new integer variable~$k_i$ that represent the cardinal of $\setexpression_i$. Our QFBAPA-formula $\phi$ is written into
 $$\underbrace{\phi[|\setexpression_i| := k_i]}_{\text{PA formula}} \land \underbrace{\lbigand_{i=1}^d k_i = |\setexpression_i| }_{\psi}.$$

Said otherwise, we have replaced $|\setexpression_i|$ by $k_i$ and enforced the equalities $k_i = |\setexpression_i|$ in a separate clause $\psi$.

\paragraph{Venn diagrams. }

Now, the main idea is to rewrite $\psi$. The goal is to get rid from $\cap$, $\cup$, etc. and only have integer variables. To do that, we will introduce integer variables $\sizeofregion{1010010}$ etc. to represent cardinals of regions. At the end $\psi$ will be replaced by $\psi'$ which is a formula free from Boolean algebra operators ($\cap$, $\cup$, etc. are deleted).

Let $\setvariable_1, \dots, \setvariable_\nbvariables$ the set variables appearing in $\setexpression_1, \dots, \setexpression_\nbbooleanexpressions$.

Each Venn diagram region is represent by a string $\coderegion \in \set{0, 1}^\nbvariables$.
\newcommand{\venndiagramintersection}{R}
The Venn diagram region corresponding to $\coderegion$ is:

$$\venndiagramintersection_\coderegion := \bigcap_{j=1}^\nbvariables \setvariable_i^{\coderegion_i}.$$

\noindent where $S_i^1 = S_i$ while $S_i^0 = \compl{S_i}$.

\begin{example}
$\venndiagramintersection_{101001} = \setvariable_1 \inter \compl{\setvariable_2} \inter \setvariable_3 \inter \compl{\setvariable_4} \inter \compl{\setvariable_5} \inter \setvariable_6$.
\end{example}

Given a set expression $\setexpression$ we can say whether a given $\venndiagramintersection_\coderegion$ is included in $\setexpression$.

\begin{example}
$\venndiagramintersection_{101001}$ is included in $S_1 \inter S_3$.
\end{example}

To check that $\venndiagramintersection_\coderegion$ is included in $\setexpression$, we can consider $\setexpression$ as a propositional formula and $\coderegion$ as a valuation. If $\coderegion$ satisfies $\setexpression$ then $\venndiagramintersection_\coderegion$ is included in $\setexpression$. In the sequel, we write $\coderegion \models \setexpression$.

\begin{example}
$\coderegion = 101001$ is the valuation 

$$\left\{
\begin{array}{l}
S_1 := 1 \\
S_2 := 0 \\
S_3 := 1 \\
S_4 := 0 \\
S_5 := 0 \\
S_6 := 1
\end{array}
\right\} $$
and it satisfies the formula $S_1 \land S_3$.
\end{example}

We introduce variable $\sizeofregion{\coderegion}$ to represent $|\venndiagramintersection_\coderegion|$. 
We rewrite $\psi$ into $\psi'$:

$$\psi' := \lbigand_{i=1}^\nbbooleanexpressions \left(k_i=\sum_{\coderegion \in \set{0, 1}^\nbvariables \suchthat \coderegion \models \setexpression_i} \sizeofregion\coderegion \right).$$

\begin{proposition}
$\phi$ is QFBAPA-satisfiable iff ${\phi[|\setexpression_i| := k_i]} \land \psi'$ is QFPA-satisfiable.
\end{proposition}

\begin{proof}
\fbox{$\Rightarrow$}
Consider a model $\modelM$ of $\phi$. We construct a model $\modelM'$ for ${\phi[|\setexpression_i| := k_i]} \land \psi'$ as follows:
\begin{itemize}
\item $\sem{k_i}{\modelM'} = |\sem{\setexpression_i}{\modelM}|$;
\item $\sem{\sizeofregion{\coderegion}}{\modelM'} = |\sem{\venndiagramintersection_\coderegion}{\modelM}|$.
\item $\sem{\integervariable}{\modelM'} := \sem{\integervariable}{\modelM}$ for the other integer variables $\integervariable$.
\end{itemize}

\fbox{$\Leftarrow$} Conversely, given a model $\modelM'$ for ${\phi[|\setexpression_i| := k_i]} \land \psi'$, we create a model $\modelM$ as follows. We create disjoint sets $R_\coderegion$ of cardinality $\sem{\sizeofregion{\coderegion}}{\modelM'}$.

\begin{itemize}
\item $\domain := \union_{\coderegion} R_\coderegion$;
\item $\sem{\setvariable_i}{\modelM} := \bigunion_{\coderegion \models \setvariable_i} R_{\coderegion}$;
\item $\sem{\integervariable}{\modelM} := \sem{\integervariable}{\modelM'}$ for the other integer variables $\integervariable$.
\end{itemize}

\end{proof}

\section{Polynomial upper bound on the number of non-zero regions}

\newcommand{\intcone}{\texttt{cone}}
Up to now, we are not finished since there are an exponential number of variables $\sizeofregion\coderegion$.
We will show that if a formula $\phi$ has a model, then there is another model in which only a polynomial number of $\sizeofregion\coderegion$ is non-zero.

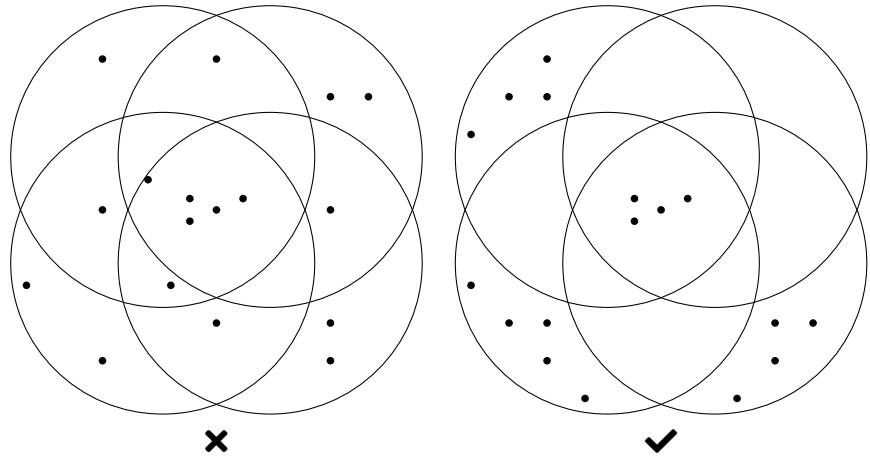
\begin{center}

\begin{tabular}{cc}
\begin{tikzpicture}[scale=1, draw=black]

  \draw[]    ( 45:1) circle (2);
  \draw[]   (135:1) circle (2);
  \draw[]  (225:1) circle (2);
  \draw[] (315:1) circle (2);

\foreach \pos in {
  (0,0),
  (0.35,0.15), (-1.5, -2), (-0.35,0.15), (-0.35,-0.15),
  (1.5, 0), (-0.9,0.4), (-1.5, 0), (1.5, -1.5),
  (2, 1.5), (0, 2), (-2.5, -1), (0, -1.5),
  (1.5, 1.5), (-1.5, 2), (-0.6,-1.0), (1.5, -2)
}
{
  \fill \pos circle (0.05);
}

\end{tikzpicture}
& 

\begin{tikzpicture}[scale=1, draw=black]

  \draw[]    ( 45:1) circle (2);
  \draw[]   (135:1) circle (2);
  \draw[]  (225:1) circle (2);
  \draw[] (315:1) circle (2);

\foreach \pos in {
  (0,0),
  (0.35,0.15), (-1.5, -2), (-0.35,0.15), (-0.35,-0.15),
  (2, -1.5), (-1.5, -1.5), (-1, -2.5), (1.5, -1.5),
  (-2.5, 1), (-1.5, 1.5), (-2.5, -1), (1, -2.5),
  (-2, 1.5), (-1.5, 2), (-2, -1.5), (1.5, -2)
}{
  \fill  \pos circle (0.05);
}

\end{tikzpicture}

\\
\faTimes & \faCheck 
\end{tabular}

\end{center}

\newcommand{\baselemma}{X}

 To do that, we will apply a Carathéodory bound for integer cones, see Th. 1 (ii) in \cite{DBLP:journals/orl/EisenbrandS06}, reformulated by the following \Cref{lemma:intcone}. We define the definition of a cone.
 
 \index{cone}
 \begin{definition}
 Given $\baselemma \subseteq \mathbb Z^d$, we define the \indexemph{integer cone} of $\baselemma$ by: 
$$\intcone(\baselemma) := \set{
\lambda_1 x_1 + \dots + \lambda_t x_t 
\suchthat t \geq 0, 
x_1, \dots, x_t \in \baselemma,
 \lambda_1, \dots, \lambda_t \in \mathbb{N}}.
$$
\end{definition}

In the following \Cref{lemma:intcone}, for a $d$-dimensional vector $x$, we write $||x||_{\infty} := \max_{i=1..d} |x_i|$. It stands for the magnitude of $x$. And then $M$ is the magnitude of a subset $X$ of vectors. The following \Cref{lemma:intcone} is a \indexemph{Carathéodory bound} saying that any vector in the cone can be obtained as a sum of only few vectors in $X$.

\index{vector}

\begin{lemma}
\cite{DBLP:journals/orl/EisenbrandS06}
\label{lemma:intcone}
Let $\baselemma \subseteq \ensZ^d$ be a finite subset. Let $M = \max_{x\in X} ||x||_{\infty}.$ \linebreak[4]
For all $b \in \intcone(\baselemma)$ there exists $\tilde{\baselemma} \subseteq \baselemma$ such that $|\tilde{\baselemma}| \leq 2d \log_2(4dM)$ and $b \in \intcone(\tilde{\baselemma})$.
\end{lemma}

\begin{proof}
Suppose that $|X| > 2d \log_2(4dM)$ (otherwise we are done). That is $M < \frac{2^{|X| / 2d}}{4d}$.

\begin{align*}
d \log (2|X|M+1) 
&< d \log \left( \frac{|X|}{2d} 2^{|X|/(2d)} + 1 \right) & \text{ by assumption} \\[6pt]
&\leq d \log \left( 2^{|X|/(2d)} \left( \frac{|X|}{2d} + 1 \right) \right) \\[6pt]
&= \frac{|X|}{2} + d \log \left( \frac{|X|}{2d} + 1 \right) \\[6pt]
&\leq \frac{|X|}{2} + d \cdot \frac{|X|}{2d} \text{ by concavity of $\log$} \\[6pt]
&= |X|.
\end{align*}

Suppose that $b = \sum_{x \in X} \lambda_x x$ with $\lambda_x \in \ensN^*$ for all $x \in X$ (all the $\lambda_x$ are strictly positive, otherwise we are done).

For $\tilde{X} \subseteq X$, we have $\sum_{x\in \tilde{X}} x \in \set{-|X|M, \dots, |X|M}^d$. So 

\begin{align*}
\card{\set{\sum_{x\in \tilde{X}} x \mid \tilde{X} \subseteq X}} 
& \leq (2|X|M + 1)^d \\
& < 2^{|X|}.
\end{align*}

So there are two sets $A, B \subseteq X$, $A\neq B$ such that 

$$\sum_{x\in A} x = \sum_{x \in B} x.$$


We set:
\begin{align*}
A' := A \setminus B \\
B := B \setminus A
\end{align*}

We have 
$$\sum_{x\in A'} x = \sum_{x\in A} x - \sum_{x \in A\inter B} x = \sum_{x\in B} x - \sum_{x \in A\inter B} x  =  \sum_{x \in B'} x.$$

W.l.o.g. we suppose that $A' \neq \emptyset$. We set $\lambda := \min_{x \in A'} \lambda_x$. Then:

\[
\begin{aligned}
b = 
\sum_{x \in X} \lambda_x x 
&= \sum_{x \in X \setminus A'} \lambda_x x + \sum_{x \in A'} \lambda_x x \\[6pt]
&= \sum_{x \in X \setminus A'} \lambda_x x + \sum_{x \in A'} (\lambda_x - \lambda)x + \lambda \sum_{x \in A'} x \\[6pt]
&= \sum_{x \in X \setminus A'} \lambda_x x + \sum_{x \in A'} (\lambda_x - \lambda)x + \lambda \sum_{x \in B'} x \\[6pt]
& = \sum_{x \in A'} (\lambda_x - \lambda)x + \sum_{x \in X \setminus (A' \union B')} \lambda_x x + \sum_{x \in B'} (\lambda_x + \lambda)
\end{aligned}
\]

The last line is another linear combination for $b$, in which all coefficients are positive but $\lambda_x - \lambda = 0$ for some $x \in A'$ by definition of $\lambda$. So we found $\tilde{X} \subsetneq X$ such that $b \in \intcone(\tilde{X})$.  We can iterate and remove elements from $X$ until $|X| \leq 2d \log_2(4dM)$.

\end{proof}

\newcommand{\subsetregionnames}{\mathbb{B}_0}
\newcommand{\subsetnonzeroregions}{\mathbb B}

From \Cref{lemma:intcone}, we can prove that we can suppose that we have at most $2d \log_2(4d)$ non-empty regions. We state the following \Cref{lemma:QFBAPAfewregions} in which we may also suppose that we know that non-empty regions are among a set $\subsetregionnames \subseteq \set{0, 1}^\nbvariables$ (at first read think $\subsetregionnames$ to be equal to $\set{0, 1}^\nbvariables$).

\begin{lemma}
Let $\subsetregionnames \subseteq \set{0, 1}^\nbvariables$.

\begin{center}
${\phi[|\setexpression_i| := k_i]} \land \lbigand_{i=1}^\nbbooleanexpressions \left(k_i = \sum_{\coderegion \in \subsetregionnames \suchthat \coderegion \models \setexpression_i} \sizeofregion\coderegion \right)$ is QFPA-satisfiable

 iff
 
  there is $\subsetnonzeroregions \subseteq \subsetregionnames$ such that $|\subsetnonzeroregions| \leq 2d \log_2(4d)$ and ${\phi[|\setexpression_i| := k_i]} \land \lbigand_{i=1}^\nbbooleanexpressions  \left(k_i =  \sum_{\coderegion \in \subsetnonzeroregions \suchthat \coderegion \models \setexpression_i} \sizeofregion\coderegion \right)$ is QFPA-satisfiable.
\end{center}
\label{lemma:QFBAPAfewregions}
\end{lemma}

\begin{proof}
\fbox{$\Uparrow$} We take a model for ${\phi[|\setexpression_i| := k_i]} \land \lbigand_{i=1}^\nbbooleanexpressions
\left( k_i= \sum_{\coderegion \in \subsetnonzeroregions \suchthat \coderegion \models \setexpression_i} \sizeofregion\coderegion\right)$ and we assign

$$\semone{\sizeofregion{\coderegion}} = 0$$

for all regions $\coderegion \in \subsetregionnames \setminus \subsetnonzeroregions$.

\fbox{$\Downarrow$}
To apply \Cref{lemma:intcone}, we rewrite $\psi'$ as the following system with $d$ equations:

\newcommand\onewhenregioninside[2]{1_{#2 \models #1}}
$$
\left\{
\begin{array}{l}
k_1 = \sum_{\coderegion\in \subsetregionnames} \sizeofregion\coderegion \onewhenregioninside{\setexpression_1}\coderegion \\
\vdots \\
k_\nbbooleanexpressions = \sum_{\coderegion\in \subsetregionnames} \sizeofregion{\coderegion} \onewhenregioninside{\setexpression_\nbbooleanexpressions}\coderegion  \\

\end{array}
\right.
$$

where $$\onewhenregioninside{\setexpression_i}\coderegion = \begin{cases} 1 \text{ if $\coderegion \models \setexpression_i$} \\
0 \text{ otherwise}
\end{cases}.$$

In a vectorial form, we get:

\newcommand\vectorK{\columnvector{k_1 \\
\vdots
\\
k_\nbbooleanexpressions}}

$$
\vectorK
=
\sum_{\coderegion \in \subsetregionnames} \sizeofregion\coderegion \columnvector{\onewhenregioninside{\setexpression_1}\coderegion \\
\vdots 
\\
\onewhenregioninside{\setexpression_\nbbooleanexpressions}\coderegion
}
$$

Said differently, if we set $\baselemma = \set{\columnvector{\onewhenregioninside{\setexpression_1}\coderegion 
\\
\vdots 
\\
\onewhenregioninside{\setexpression_\nbbooleanexpressions}\coderegion
}
 \suchthat \coderegion \in \subsetregionnames}$, we get:

$$\vectorK
\in \intcone(\baselemma).$$

By \Cref{lemma:intcone}, there exists a subset  $\mathbb{B} \subseteq \subsetregionnames$ of size at most $2d\log_2(4d)$ and values for the $\sizeofregion\coderegion$ such that

$$
\vectorK = 
\sum_{\coderegion \in \mathbb{B}} \sizeofregion\coderegion \columnvector{\sem{\setexpression_1}\coderegion \\
\vdots 
\\
\sem{\setexpression_\nbbooleanexpressions}\coderegion
}.
$$

Just to recall:

\begin{center}
\begin{tabular}{|l|l|}
\hline
$d$ & number of set expressions \\
\hline
$e$ & number of set variables \\
\hline
\end{tabular}
\end{center}

\end{proof}

\section{QBFPAPA satisfiability in NP}

\index{NP}
Here is an algorithm for testing the satisfiability problem of QFBAPA-formula $\phi$.

\begin{center}
\begin{algo}
input: a QFBAPA-formula $\phi$

output: has a non-rejecting execution iff $\phi$ is QFBAPA-satisfiable

\begin{algoblocprocedure}{QFBAPAsat($\phi$)}

$d$ := number of Boolean set expressions

$e$ := number of set variables

Guess a subset $\mathbb{B} \subseteq \set{0, 1}^e$ of size $2d\log_2(4d)$

Check whether the QFPA-formula ${\phi[|\setexpression_i| := k_i]} \land \lbigand_{i=1}^d \sum_{\coderegion \in \mathbb{B} \suchthat \coderegion \models \setexpression_i} \sizeofregion\coderegion = k_i $ is satisfiable

\end{algoblocprocedure}
\end{algo}
\end{center}

\begin{example}
Consider the formula $(|S \inter T | \leq 5) \land (|S| > |T|)$. We have $e = 2$ set variables: $S$ and $T$. We have $d = 3$ set expressions: $S\inter T$, $S$ and $T$.
\end{example}

\begin{theorem}
QFBAPA satisfiability problem is in NP.
\end{theorem}

\begin{proof}
We have to prove that QFBABAsat is sound and complete. We have $\phi$ QFBAPA-satisfable iff $\phi[|\setexpression_i := k_i]\land \lbigand_{i=1}^\nbbooleanexpressions 
\left( k_i=
\sum_{\coderegion \in \set{0, 1}^\nbvariables \suchthat \coderegion \models \setexpression_i} \sizeofregion\coderegion
\right)$ is QFPA-satisfiable. 

\fbox{$\Rightarrow$} If $\phi$ has a model, then $\phi[|\setexpression_i := k_i]\land \lbigand_{i=1}^\nbbooleanexpressions 
\left(k_i = 
\sum_{\coderegion \in \set{0, 1}^\nbvariables \suchthat \coderegion \models \setexpression_i} \sizeofregion\coderegion\right)$ has a model $\modelM$: interpret $k_i$ as the cardinality of $\setexpression_i$, i.e. $\sem{\setexpression_i}\modelM$ and $\sizeofregion\coderegion$ as the cardinality of the region $R_\coderegion$, i.e. $\sem{R_\coderegion}\modelM$.

\fbox{$\Leftarrow$} If $\phi[|\setexpression_i := k_i]\land \lbigand_{i=1}^\nbbooleanexpressions 
\left(k_i = 
\sum_{\coderegion \in \set{0, 1}^\nbvariables \suchthat \coderegion \models \setexpression_i} \sizeofregion\coderegion\right)$ is satisfiable in a QFPA-model $\modelM$. 
We construct a QFBAPA model as follows. First, we construct regions $R_\coderegion$ with~$\sizeofregion\coderegion$ each. The domain is the (disjoint) union of these regions. Then we interpret $\setexpression_i$ as the union of $R_\coderegion$ such that $\coderegion \models \setexpression_i$ while the interpretation of integer variables remain the same as in $\modelM$.

Finally, the procedure QFBAPAsat is non-deterministic algorithm that runs in polynomial time in~$|\phi|$.
\end{proof}

\section{Application: PSPACE Tableau Method for $K^\#$}
\index{tableau method}
\index{PSPACE}

\newcommand{\algosatKsharp}{\textsf{sat$K^\#$}}
We write a non-deterministic procedure inspired from \cite{DBLP:journals/corr/abs-2307-05150} (in which we forgot to use QFBAPA) and \cite{DBLP:conf/frocos/Baader17} (in which QFBAPA is used but for a description logic close to $K^\#$).

\subsection{Description of the algorithm}

\newcommand{\setinequalities}{\mathcal S}
\begin{center}
\begin{algo}
input: $\Gamma$ a set of $K^\#$-formulas

output: has a non-rejecting execution iff $\Gamma$ is $K^\#$-satisfiable

\begin{algoblocprocedure}{\algosatKsharp($\Gamma$)}

\begin{algoblocfor}{non-nested expressions $1_\phi$}
choose either to add $\phi$ in $\Gamma$ and replace $1_\phi$ by 1 in $\Gamma$

~~~~~~~~~~ or to add $\lnot \phi$ in $\Gamma$ and replace $1_\phi$ by 0 in $\Gamma$
\end{algoblocfor}

apply non-deterministically Boolean tableau rules to $\Gamma$

let $\setinequalities$ be the set of inequalities in $\Gamma$

let $\# \psi_1, \dots \# \psi_\nbbooleanexpressions$ be a list of all constructions of the form $\# \psi$ appearing in $\Gamma$

Guess $\mathbb{B} \subseteq \set{0, 1}^d$ of size $\leq 2d\log_2(4d)$

Replace in $\setinequalities$ each occurrence of $\# \psi_i$ by $\sum_{\coderegion \in \mathbb B \suchthat \coderegion_i = 1} \sizeofregion\coderegion$


We obtain $\setinequalities'$

Check that $\setinequalities'$ is QFPA-satisfiable

\begin{algoblocfor}{$\coderegion \in \mathbb B$}

\algosatKsharp($\set{\lnot \psi_i \suchthat \coderegion_i = 0} \union \set{\psi_i \suchthat \coderegion_i = 1}$)

\end{algoblocfor}

\end{algoblocprocedure}
\end{algo}
\end{center}

In the algorithm, at each step, we extract the set of inequalities in $\Gamma$. For instance, we may get $S$ to be

$$\left\{ \begin{array}{l}
\# \psi_1 + 3\# \psi_2 \leq 5 \\
2 \# \psi_1 + 4 \# \psi_3 \leq 42 
\end{array}
\right. $$

The use we make of QFBAPA is "trivial" since set expressions and set variables coincide: they are $\psi_1, \psi_2, \dots$. The content of $\psi_i$ may be Boolean (and also modal) but the Boolean reasoning is directly handled by tableau rules. So the use of the QBFPAPA technique is for $d = e$ (we write $\mathbb{B} \subseteq \set{0, 1}^e$ instead of $\mathbb{B} \subseteq \set{0, 1}^d$).

We then guess $\mathbb B$ which are the non-zero regions. For instance, having

$$\mathbb B := \set{100, 101, 111}$$

mean that we are trying to create successors for the current vertex where
\begin{enumerate}
\item $\psi_1 \land \lnot \psi_2 \land \lnot \psi_3$ holds in a subset of successors ($\sizeofregion{100}$ is the cardinal);
\item $\psi_1 \land \lnot \psi_2 \land  \psi_3$ holds  in a subset of successors ($\sizeofregion{101}$ is the cardinal);
\item $\psi_1 \land  \psi_2 \land  \psi_3$ holds  in a subset of successors ($\sizeofregion{111}$ is the cardinal).
\end{enumerate}

Other combinations (e.g. $\lnot \psi_1 \land  \psi_2 \land  \psi_3$) are not present in the model we construct.

\newcommand\bloup[1]{$\sizeofregion{#1}$ (if #1 in $\mathbb B$)}
We then replace $\# \psi_1$ by the sum of \bloup{100}, \bloup{101}, \bloup{110}, \bloup{111}. Similarly for $\# \psi_2$ and $\# \psi_3$.

As we are going then to check for satisfiability of 1.-3. (for loop at the end of  the algorithm), we add that $\sizeofregion\coderegion \geq 1$ for all non-zero regions $\coderegion$.

\subsection{Soundness and completeness}

\begin{proposition}
If $\algosatKsharp(\Gamma)$ has an accepting execution then $\Gamma$ is $K^\#$-satisfiable.
\end{proposition}

\begin{proof}
We consider an accepting execution.
    
    We construct the root to be a vertex $u$ labelled by $\labellinginitial(u)$ constrained by the presence/absence of $x_i = 1$ in $\Gamma$.
    
    By induction we know that for each non-region $\coderegion \in \mathbb{B}$, we know that $\psi_{\coderegion} := \lbigand_{i \suchthat \coderegion_i = 0} \lnot \psi_i \land \lbigand_{i \suchthat \coderegion_i = 1} \psi_i$ is \Ksharp-satisfiable in a model $G_\coderegion, u_\coderegion$. As $S'$ is satisfiable, there is a QFPA-model $\semone{\cdots}$ for $S'$. 
    
    We make $\semone{\sizeofregion\coderegion}$ copies of $G_\coderegion, u_\coderegion$ that we declare as successors of $u$.
	 We obtain a pointed graph satisfying $\Gamma$.
\end{proof}

\begin{proposition}
If $\Gamma$ is $K^\#$-satisfiable, then $\algosatKsharp(\Gamma)$ has an accepting execution.
\end{proposition}

\begin{proof}
Consider a pointed graph $G, u$ which is a model of $\Gamma$.
 We apply the Boolean tableau rules accordingly for $\lor$, $1_\phi$ (e.g. choosing to put $\phi$ in $\Gamma$ if $\phi$ holds in $G, u$ and 0 otherwise). 
 As $\Gamma$ holds in $G, u$,
  $\setinequalities$ is QFBAPA-satisfiable: let $\modelM$ be a model of $\setinequalities$. Let $\subsetregionnames$ be the subset of regions that are non-empty in $\subsetregionnames$. 
 By \Cref{lemma:QFBAPAfewregions}, there is a subset $\mathbb{B} \subseteq \subsetregionnames$ with $|\mathbb{B}| \leq 2d \log_2(4d)$ such that there is a model $\modelM'$ in which $\setinequalities$ is satisfied and $\mathbb B$ is the exactly the set of non-zero regions. We consider the execution that guesses exactly $\mathbb B$.
 As $\mathbb B \subseteq \subsetregionnames$, we know that $\psi_\coderegion$ is satisfiable for all $\coderegion \in \mathbb B$.
 By induction there an accepting execution.
\end{proof}

\section*{Open questions}

\begin{itemize}
\item Efficient implementation
\item Nice proof system
\item Write down a comprehensive proof for the satisfiability problem of $\Ksharp$ for undirected graphs (and other classes of graphs)
\end{itemize}

\newpage
\section*{Exercises}

\begin{exercise}
Take a $K^\#$-formula of your choice and apply the algorithm $\algosatKsharp$.
\end{exercise}

\begin{exercise}
Show that QFBAPA is NP-hard even if numbers that are written are in formulas are 0 and 1.
\end{exercise}

\begin{exercise}
Prove formally the theorems of the chapter.
\end{exercise}

\begin{exercise}
Adapt the algorithm $\algosatKsharp$ when we are search for an undirected graph.
\end{exercise}

\chapter{Other settings}
\label{chapter:globalreadout}
%
%
\newcommand\equationbox[1]{\colorbox{blue!5!white}{\ensuremath{#1}}}
\newcommand{\agreggationfunctiong}{\agreggationfunction^g}

In this section, we reproduce some proofs of \cite{DBLP:conf/icalp/BenediktLMT24} with our $\Ksharp$-formalism, and also using our language of GNN expressions.

\section{Global readout}
\index{global readout}

\newcommand{\sharpU}{\#^g}
\newcommand{\Ksharpglobal}{K^{\#,\#^g}}
\newcommand{\setequations}{\mathcal E}
\newcommand{\boxU}{\lbox^g}

We consider now GNNs with global readout, meaning that the update is:

\begin{align*}
\aGNNoutlayer{t}(G,u) = & \vec \activationfunction(\weightcolor{A_t} \times \aGNNoutlayer{t-1}(G,u) \\
 & + \weightcolor{B_\timestep} \times \sum \multiset{\aGNNoutlayer{t-1}(G,v) \suchthat v \in E(u)} \\
 & + \weightcolor{C_\timestep} \times \sum \multiset{\aGNNoutlayer{t-1}(G,v) \suchthat v \in V} \\
 & + \weightcolor{b_\timestep}).
\end{align*}

The term $\weightcolor{C_\timestep} \times \sum \multiset{\aGNNoutlayer{t-1}(G,v) \suchthat v \in V}$ is the global readout term.

To capture global readout, we extend $\Ksharp$ with a universal operator $\sharpU \phi$ which is the global counting in the whole graph whose semantics is
\begin{align*}
\sem{\sharpU \phi}{G, u}  :=  |\sem{\phi}G|
\end{align*}

The obtained logic is $\Ksharpglobal$. In this logic, we can simulate the \indexemph{global modality} $\boxU$ as follows:
\begin{align*}
\boxU \phi ::= \sharpU(\phi) = \sharpU(\top)
\end{align*}

\noindent
which says that $\phi$ holds in all vertices. This logic captures GNNs with global readout in the same way than $\Ksharp$ captures GNNs. It means that we can restate \Cref{proposition:KsharptoGNN} and \Cref{proposition:truncrelugnnstoKsharp} for $\Ksharp$ and GNN-expressions with a new construction $\agreggationfunctiong$ for the global readout aggregation.

\index{NEXPTIME}
\begin{theorem}(see Theorem 23 in \cite{2510.08045})
The satisfiability problem of $\Ksharpglobal$ in directed graphs is NEXPTIME-complete.
\end{theorem}

\index{undecidability}

\begin{theorem}(adapted from Theorem 4.16, \cite{DBLP:conf/icalp/BenediktLMT24})
The satisfiability problem of $\Ksharpglobal$ in undirected graphs is undecidable.
\end{theorem}

\begin{proof}
We recall Hilbert's tenth problem \cite{DBLP:conf/cie/Matiyasevich05}. The following problem is undecidable:
\begin{itemize}
\item input: a set $\setequations$ of equations of the form $x = 1$, $x = x' + x''$, or $x = x' \times x''$;
\item output: yes if $\setequations$ has a solution in $\ensN$. 
\end{itemize}

\renewcommand{\equation}{e}
We reduce from the Hilbert's tenth problem.
We are going to build a $\Ksharpglobal$-formula $tr(\setequations)$ in poly-time in the number of bits to represent $\setequations$ such that $\setequations$ has a solution in $\setN$ iff $tr(\setequations)$ is satisfiable.

\paragraph{Intuition.}
The idea is to constrain a set (a graph!) where vertices are tagged by equations \emph{and} variables. The subset of vertices tagged by equation $\equation$ and variable $x$ is denoted by $V_{e,x}$. The $V_{e,x}$ should be a partition of the set $V$ of all vertices:

$$V = \bigsqcup_{e\in\setequations, \text{variables }x} V_{e,x}.$$

More importantly, for all $x$, $|V_{e,x}|$ should not depend on $e$ and should be the value of $x$ in a solution of $\setequations$! The subset $V_{e,x}$ are indexed by $e$ in order to handle the constrain of each equation $e$ separately.

\newcommand{\itisvariable}[1]{p_{#1}}
\newcommand{\itisequation}[1]{p_{#1}}
\newcommand\itis[2]{\itisequation{#1} \land \itisvariable{#2}}

We introduce the following propositions: $\itisequation{e}$ that says that a vertex is tagged by equation $e$, and $\itisvariable x$ that says that a vertex is tagged by variable $x$. Intuitively, vertices in $V_{e,x}$ are those in which both $\itisequation{e}$ and $\itisvariable x$ hold.

We now define $tr(\setequations)$ that does the job.
First:
\begin{align}
\lbigand_{e, e' \suchthat e \neq e'} \boxU\lnot (\itisequation{e} \land \itisequation{e'}) \\
\lbigand_{x, x' \suchthat x \neq x'} \boxU\lnot (\itisvariable{x} \land \itisvariable{x'}) \\
\lbigand_{x} \lbigand_{e, e'} \sharpU(\itis e x) = \sharpU(\itis {e'} x)
\end{align}

Up to now, at most one equation tag, at most one variable tag and $|V_{e,x}|$ only depend on~$x$. In order to emphasize that point, we write $|V_{\bullet,x}|$ instead of $|V_{e,x}|$.

We now introduce a formula $\phi_e$ to simulate the semantics of each equation $e$. For all equations $e$.

\begin{itemize}
\item \fbox{$e = \equationbox{x = 1}$}
We set $\phi_e := \sharpU(\itis{e}{x}) = 1$ to intuitively say $|V_{e,x}| = |V_{\bullet,x}| = 1$.

\item \fbox{$e = \equationbox{x = x' + x''}$}
We set $$\phi_e := \boxU\left(
\begin{array}{ll}
([\itisequation{e} \land (\itisvariable{x'} \lor \itisvariable{x''})] \limply \#(\itis{e}{x})=1) \land \\
 ([\itis {e} {x}] \limply \#(\itisequation{e} \land (\itisvariable{x'} \lor \itisvariable{x''}))=1)
\end{array}\right)$$

The formula $\phi_e$ says that each vertex in $V_{e, x'} \sqcup V_{e, x''}$ has a single neighbour in $V_{e,x}$. Conversely, each vertex in $V_{e,x}$ has a single neighbour in  $V_{e, x'} \sqcup V_{e, x''}$. Said differently, there is a bijection between $V_{e, x'} \sqcup V_{e, x''}$ and $V_{e,x}$.

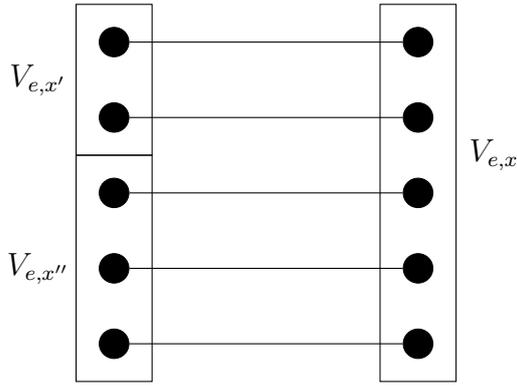
\begin{center}
\begin{tikzpicture}
 \tikzstyle{vertex} = [draw, circle, fill=black];

   \draw (1, 0) rectangle (2, 2);
   \draw (1, -3) rectangle (2, 0);
   \draw (5, -3) rectangle (6, 2);
   \node[vertex] (1) at (1.5, 1.5) {};
   \node[vertex] (2) at (1.5, 0.5) {};
   \node[vertex] (3) at (1.5, -0.5) {};
   \node[vertex] (4) at (1.5, -1.5) {};
   \node[vertex] (5) at (1.5, -2.5) {};
   
  \node[vertex] (1') at (5.5, 1.5) {};
   \node[vertex] (2') at (5.5, 0.5) {};
   \node[vertex] (3') at (5.5, -0.5) {};
   \node[vertex] (4') at (5.5, -1.5) {};
   \node[vertex] (5') at (5.5, -2.5) {};
   
   \foreach \i in {1, 2, ..., 5} {
       \draw (\i) -- (\i');
   }
   
   \node (A) at (0.5, 1) {$V_{e, x'}$};
   \node (A) at (0.5, -1.5) {$V_{e, x''}$};
   \node (A) at (6.5, 0) {$V_{e, x}$};
\end{tikzpicture}
\hfill $5 = 2 \times 3$
\end{center}

So $|V_{\bullet,x}| = |V_{\bullet,x'}| + |V_{\bullet,x''}|$ which appropriately simulate the equation $x = x' + x''$.

\item \fbox{$e = \equationbox{x = x' \times x''}$}
We set
$$\phi_e := \boxU\left(
\begin{array}{ll}
([\itis e {x'}\limply \#(\itis e x) = \sharpU(\itis e {x''})] \land \\
 ([\itis {e} {x}] \limply \#(\itis{e} {x'})=1)
\end{array}\right)$$

The formula $\phi_e$ says that each vertex in $V_{e, x'}$ has $|V_{e, x''}|$ neighbours in $V_{e,x}$. Conversely, each vertex in $V_{e,x}$ has one neighbour in $V_{e, x'}$. So there is a $|V_{e, x''}|$-to-one mapping from $V_{e, x'}$ to $V_{e,x}$.

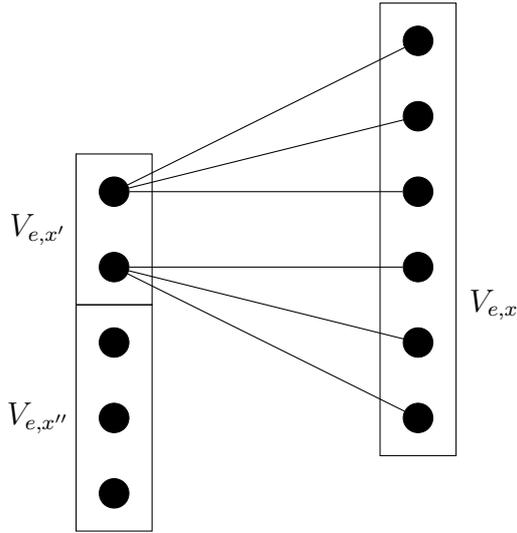
\begin{center}
\begin{tikzpicture}
 \tikzstyle{vertex} = [draw, circle, fill=black];

   \draw (1, 0) rectangle (2, 2);
   \draw (1, -3) rectangle (2, 0);
   \draw (5, -2) rectangle (6, 4);
   \node[vertex] (1) at (1.5, 1.5) {};
   \node[vertex] (2) at (1.5, 0.5) {};
   \node[vertex] (3) at (1.5, -0.5) {};
   \node[vertex] (4) at (1.5, -1.5) {};
   \node[vertex] (5) at (1.5, -2.5) {};
   
  \node[vertex] (1') at (5.5, 3.5) {};
   \node[vertex] (2') at (5.5, 2.5) {};
   \node[vertex] (3') at (5.5, 1.5) {};
   \node[vertex] (4') at (5.5, 0.5) {};
   \node[vertex] (5') at (5.5, -0.5) {};
   \node[vertex] (6') at (5.5, -1.5) {};
   
   \foreach \i in {1, 2, 3} {
       \draw (1) -- (\i');
   }
   
    \foreach \i in {4,5,6} {
       \draw (2) -- (\i');
   }
   \node (A) at (0.5, 1) {$V_{e, x'}$};
   \node (A) at (0.5, -1.5) {$V_{e, x''}$};
   \node (A) at (6.5, 0) {$V_{e, x}$};
\end{tikzpicture}
\hfill $6 = 2 \times 3$
\end{center}

So $|V_{\bullet,x}| = |V_{\bullet,x'}| \times |V_{\bullet,x''}|$ which appropriately simulate the equation $x = x' \times x''$.
\end{itemize}

The reduction is given by:
$$tr(\setequations) := (7.1) \land (7.2) \land (7.3) \land \lbigand_{e \in \setequations} \phi_e.$$
It is left to prove formally that $\setequations$ has a solution in $\setN$ iff $tr(\setequations)$ is satisfiable. We left the end of the proof to the reader.
\end{proof}

\begin{corollary}
The satisfiability problem of a GNN with TruncReLU as an activation function and with global readout is undecidable.
\end{corollary}

\begin{proof}
By the previous and a similar translation from $\Ksharpglobal$ into GNNs with TruncReLU and global readout.
\end{proof}

\section{ReLU}
\index{ReLU}

\begin{theorem}
Verification of GNNs with global readout and ReLU as an activation function on undirected graphs is undecidable.
\end{theorem}

\begin{proof}
Because we can simulate truncReLU with ReLU. It is was already undecidable for truncReLU.
\end{proof}

\begin{theorem}
Verification of GNNs with global readout and ReLU as an activation function on directed graphs is undecidable.
\end{theorem}

\begin{proof}
\newcommand{\countforvar}[1]{count_{#1}}
Reduction from \indexemph{Hilbert's tenth problem}. We still only the fact truncReLU can be similated by ReLU: this enables to express logical properties. But we also need ReLU for storing arbitrary long integers throw the computation.

The idea is to tag the vertices with a feature $\countforvar x$. When $\countforvar x = 0$ the vertex does not count and $\countforvar x = 1$ it means that the vertex counts for 1 in the value of variable $x$. We interpret the value of $x$ by $\sharpU \countforvar x$. More precisely:

$$\sharpU \countforvar x := \agreggationfunctiong (ReLU (\countforvar x)).$$

\begin{itemize}
\item \fbox{$e = \equationbox{x = 1}$} $\phi_e := \sharpU \countforvar x = 1$. We can write equalities by inequalities and write inequalities as for the translation from $\Ksharp$ into GNN with truncReLU (and we simulate truncReLU by ReLU).
\item \fbox{$e = \equationbox{x = x' + x''}$}  $\phi_e := \sharpU \countforvar x = \sharpU \countforvar {x'} + \sharpU \countforvar {x''}$.
\item \fbox{$e = \equationbox{x = x' \times x''}$} We cannot write multiplication constraint directly. So introduce extra features $y_e$ and $y''_e$. 
\begin{align}
\phi_e := & ((y_e = 0) \land (z_e = 0)) \lor ((y_e = \sharpU\countforvar {x'}) \land (z_e = 1))  \\
 & \land \agreggationfunctiong (y_e) = \sharpU\countforvar x \\
  & \land \agreggationfunctiong (z_e) = \sharpU\countforvar{x''}
\end{align}

The objective of (7.4) is to replace each 1-value of $z_e$ by $\sharpU\countforvar {x'}$ so that $y_e = \sharpU\countforvar {x'} \times z_e$ at each vertex. By taking the global aggregation over all vertices, we get:

$$\sharpU \countforvar x = \sharpU\countforvar{x'} \times \sharpU\countforvar{x''}.$$

\end{itemize}

Finally the GNN is: $\lbigand_{e \in \setequations} \phi_e$. We left the equivalence to the reader.
\end{proof}

\begin{theorem}
Verification of GNNs with ReLU as an activation function on undirected graphs is undecidable.
\end{theorem}

\begin{proof}
We simulate the global modality by a local one around the root.
\end{proof}

%
%
%


\bibliographystyle{apalike}
\bibliography{biblio}

\printindex

\end{document}